\pdfoutput=1
\documentclass[journal,10pt]{IEEEtran}  
\usepackage{cite}
\usepackage{indentfirst}
\usepackage{graphicx}
\usepackage{changepage}
\usepackage{stfloats}
\usepackage{amsfonts,amssymb}
\usepackage{bm}
\usepackage{algorithm}
\usepackage{algorithmic}
\usepackage{amsmath}
\usepackage{amsmath,amssymb,amsfonts}
\usepackage{times}
\usepackage{url}
\usepackage{cite}
\usepackage{textcomp}
\usepackage{xcolor}
\usepackage{color}
\usepackage{setspace}
\usepackage{enumitem}
\usepackage{float}
\usepackage{diagbox}
\usepackage[T1]{fontenc}
\usepackage[utf8]{inputenc}
\usepackage{authblk}
\usepackage{threeparttable}
\usepackage{booktabs}
\usepackage{footnote}
\usepackage{graphicx}
\usepackage{pifont}
\usepackage{subfigure}
\usepackage{mathrsfs}
\allowdisplaybreaks[4]

\newenvironment{proof}{{\indent  \indent \it Proof:}}{\hfill $\blacksquare$}

\begin{document}
\title{Cooperative Cellular Localization with Intelligent Reflecting Surface: Design, Analysis and Optimization}

\author{
	Kaitao Meng, \textit{Member, IEEE}, Qingqing Wu, \textit{Senior Member, IEEE}, Wen Chen, \textit{Senior Member, IEEE}, and Deshi Li
	\thanks{K. Meng is with the State Key Laboratory of Internet of Things for Smart City, University of Macau, Macau, 999078, China, and also with the Department of Electrical and Electronic Engineering, University College London, WC1E 7JE, UK (email: kaitao.meng@ucl.ac.uk). Q. Wu and W. Chen are with the Department of Electronic Engineering, Shanghai Jiao Tong University, Shanghai 201210, China (emails: \{qingqingwu, wenchen\}@sjtu.edu.cn). D. Li is with the Electronic Information School, Wuhan University, Wuhan, 430072, China (email: dsli@whu.edu.cn). }
	\thanks{The work of K. Meng was supported in part by UKRI under Grant EP/Y02785X/1. The work of Q. Wu was supported by NSFC 62371289, NSFC 62331022, and Guangdong science and technology program under grant 2022A0505050011. The work of W. Chen was supported in part by NSFC 62071296, Shanghai 22JC1404000, 20JC1416502, and PKX2021-D02. The work of D. Li was supported in part by the National Key R\&D Program of China (No. 2023YFE0206600). }   
	\thanks{An earlier version of this paper was presented in part at the 2023 IEEE Global Communications Conference (GLOBECOM) Workshop \cite{Meng2023Vehicle}.}
}
\maketitle

\begin{abstract}
Autonomous driving and intelligent transportation applications have dramatically increased the demand for high-accuracy and low-latency localization services. While cellular networks are potentially capable of target detection and localization, achieving accurate and reliable positioning faces critical challenges. Particularly, the relatively small radar cross sections (RCS) of moving targets and the high complexity for measurement association give rise to weak echo signals and discrepancies in the measurements. To tackle this issue, we propose a novel approach for multi-target localization by leveraging the controllable signal reflection capabilities of intelligent reflecting surfaces (IRSs). Specifically, IRSs are strategically mounted on the targets (e.g., vehicles and robots), enabling effective association of multiple measurements and facilitating the localization process. We aim to minimize the maximum Cramér-Rao lower bound (CRLB) of targets by jointly optimizing the target association, the IRS phase shifts, and the dwell time. However, solving this CRLB optimization problem is non-trivial due to the non-convex objective function and closely coupled variables. For single-target localization, a simplified closed-form expression is presented for the case where base stations (BSs) can be deployed flexibly, and the optimal BS location is derived to provide a lower performance bound of the original problem. Then, we prove that the transformed problem is a monotonic optimization, which can be optimally solved by the Polyblock-based algorithm. Moreover, based on derived insights for the single-target case, we propose a heuristic algorithm to optimize the target association and time allocation for the multi-target case. Furthermore, we provide useful guidance for the practical implementation of the proposed localization scheme by theoretically analyzing the relationship between time slots, BSs, and targets. Simulation results verify that deploying IRS on vehicles and effective phase shift design can effectively improve the resolution ability of multi-vehicle positioning and reduce the requirements of the number of BSs. 
\end{abstract}

\begin{IEEEkeywords}
	Cooperative localization, intelligent reflecting surfaces, Cramér-Rao lower bound, integrated sensing and communication, target association, phase shift design.
\end{IEEEkeywords}
\newtheorem{thm}{\bf Lemma}
\newtheorem{remark}{\bf Remark}
\newtheorem{Pro}{\bf Proposition}
\newtheorem{theorem}{\bf Theorem}
\newtheorem{Assum}{\bf Assumption}
\newtheorem{Cor}{\bf Corollary}

\section{Introduction}

In recent years, applications such as autonomous driving and intelligent transportation have brought a significant increase in the need for high-accuracy and low-latency localization services for mobile platforms, e.g., robots and vehicles \cite{Challenges2021Gyawali}. In general, localization techniques can be divided into two main categories: device-based sensing and device-free sensing \cite{Liu2022Survey}. For device-based sensing, it consists of targets transmitting and/or receiving signals, which generally requires strict time synchronization between transmitter and receiver, e.g., global positioning systems (GPS) and wireless sensor networks \cite{Win2018Network}, whereas device-free sensing enable analysis of the echo signals reflected from the target without the involvement of the target transmitter or receiver, such as radar \cite{richards2014fundamentals}. Also, the localization performance of GPS is generally affected by blockages, especially in densely obstructed environments. Hence, it is more practical and efficient for radar sensing to localize targets in high-density and high-mobility vehicular scenarios \cite{Huang2021MIMORadar, Bayesian2021Yuan}. Dedicated radar equipment is primarily used for military purposes and is not widely deployed in urban environments. Besides radar equipment,  cellular networks are potentially capable of detecting targets and measuring ranges, benefiting from the continuous increase in the communication frequency band. Additionally, cellular networks generally have a large coverage area and are widely available in cities. Thus, the use of  cellular networks to improve target localization performance has received increasing attention. 

In the literature, the majority of work on cellular network-aided positioning focuses on the deployment of base stations (BSs), beamforming design, and allocation of resources \cite{Xu2019Optimal, shi2022device, wang2022multi, Shi2021JointAssignment}. For example, under the single-target setup, an optimal sensor node placement strategy for circular time-of-arrival (TOA) localization in the three-dimensional (3D) space was proposed in \cite{Xu2019Optimal}. For the multi-target case, in \cite{shi2022device}, the orthogonal frequency-division multiplexing (OFDM) channel estimation method was adopted to obtain the corresponding delay of all bidirectional BS-target-BS paths, which was then used to jointly optimize data association and target localization. A high level of synchronization between BSs is required for this method, and twice as many BSs as the target number must be involved for effective positioning. Furthermore, \cite{wang2022multi} presented a target association method based on tracking techniques, where the similarity of the distributions of estimated and predicted locations were compared to distinguish different vehicle. The above studies have proved that wireless networks are capable of providing localization services by analyzing the echo signal received at the BSs.

Despite their advantages in localization, cellular networks still face some crucial challenges to achieving high-precision and ultra-reliable positioning for targets with high mobility. First, since receivers are unaware of concurrent reflections from multiple targets, it is difficult to match measurements (e.g., distances and angles) with the right target \cite{kazemi2021data}. For instance, in conventional distance-based multi-target localization schemes \cite{shi2022device, yang2022multitarget}, the association between multiple targets and extracted measurements needs to be determined based on the location of intersections calculated from several measured distances. In this case, the number of intersections inevitably increases exponentially with respect to the number of targets, making it computationally expensive to determine which intersections are real targets or "ghost/shadow" targets \cite{bosse2015direct, Yi2020SuboptimalRadar}. Second, radar cross sections (RCS) of served moving targets (e.g., vehicles) in urban environments are generally small, and the power of the transmitter may be limited, thus leading to weak echo signals and failure of precise target detection/tracking \cite{Meng2023Intelligent}. Finally, when the trajectories of several targets are in close proximity to each other, even with sufficient BSs, conventional localization methods have difficulty in discriminating these targets. It is worth noting that these issues are primarily caused by the uncontrollability of the echo signals reflected from targets. However, with the development of metamaterial technology, it has become possible to control the propagation path of electromagnetic waves via intelligent reflecting surfaces (IRSs) \cite{Wu2022IntelligentSurfaceOverview}. As a result, controlling echo signals is a promising solution to enhance localization performance and reduce interference with other wireless systems. To the best of our knowledge, the idea of utilizing echo signal control for enhancing multi-target localization performance in cellular networks is seldom explored. This thus motivates us to design a new cooperative localization scheme with the exploitation of IRSs.

Using massive low-cost reflecting elements, IRSs can achieve larger communication coverage and improved transmission quality \cite{wu2019intelligent, Chen2023IRSAidedOffloading}. In addition to providing effective communication services, IRSs can also aid to sense targets in blind areas by establishing virtual light-of-sight (LoS) links between BSs and targets \cite{Meng2022Intelligent, Dardari2022LOSNLOS}. For instance, the authors in \cite{Meng2022Intelligent} proposed a novel sensing scheme to sense multiple targets simultaneously by establishing a relationship between the target directions and signature sequences (SSs), and thus the target directions can be obtained by analyzing echo signals modulated with SSs. In \cite{Hua2023Intelligent}, receive sensors are deployed on the IRS surfaces to reduce echo signal path loss and facilitate the of target parameters. Moreover, IRSs can also further promote the mutual benefits between sensing and communication tasks \cite{meng2023Surfaceintelligent}. At present, the majority of existing works have focused on exploiting IRSs to improve communication/sensing performance among static or low-mobility users/targets, where IRSs are deployed in fixed locations, e.g., buildings and billboards \cite{Simultaneously2022Mu, Xu2022Channel}. However, the sensing coverage ability under such IRS deployment strategies may be constrained for the localization of high-mobility targets. How to achieve high-resolution and high-reliable localization as well as reduce interference between different measurements in high-mobility scenarios is still an open issue.

In this paper, we propose a novel multi-target localization scheme by employing IRSs on the surfaces of vehicles to facilitate target association and localization, as shown in Fig. \ref{figure1}.\footnote{The proposed localization can also be utilized to improve localization performance for other types of targets mounted with IRSs, such as robots and unmanned aerial vehicles (UAVs), as shown in Fig.~\ref{figure1a}.} In this way, besides the resource allocation on the transmitter sides, the controllable signal reflection at the IRS/target sides can be exploited to enhance the reflected signal strength and reduce interference with unassociated devices. Furthermore, to eliminate ranging bias from various sensing sites, especially when dealing with extended targets, the IRS at the target location can be treated as a point target, thereby avoiding the complex algorithm for tracking extended targets.
Then, in our considered system, the localization performance is described based on the Cramér-Rao lower bound (CRLB) \cite{Li2008RangeCompression, Liu2022CramerRaoBound}. Due to the ability to actively control echo signals, the association between BSs and IRSs is optimized together with the dwell time on each target in this work, thereby ensuring the fairness of multiple targets. Notice that there is a fundamental trade-off for association optimization between BSs and IRSs. On one hand, within a given duration, each IRS tends to establish association with more BSs to increase the diversity gain \cite{sadeghi2021target}. On the other hand, more associations inevitably reduce the dwell time of each IRS, resulting in decreased measurement accuracy and greater interference between different BSs. Thus, it is important to properly optimize time allocation and target association to provide a better localization performance improvement.

To investigate the localization performance of the proposed cooperative sensing scheme, the task duration is divided into several time slots to facilitate target association. Then, we aim to minimize the maximum CRLB of targets by jointly optimizing the target association, the IRS phase shifts, and the dwell time. However, solving this CRLB optimization problem is highly non-trivial due to the non-convex objective function, closely coupled variables, and the uncertain number of time slots. To handle this issue, for single-target localization, we first derive a performance lower bound with flexible BS deployments, to provide insights for algorithm design. Furthermore, for multi-target localization, we theoretically analyzed the relationship between performance and system parameters, providing useful guidance for practical implementation. The main contributions of this paper can be summarized as follows:
\begin{itemize}[leftmargin=*]
	\item We propose a novel IRS-aided multi-target localization scheme for vehicular networks, where an IRS is deployed on vehicles to enhance localization performance by jointly optimizing the target association, the phase shift vectors, and time resources. Given the uncertain prior location of targets, the beam flattening technique is adopted to ensure the reliability and effectiveness of distance measurement.
	\item We derive a closed-form expression of CRLB, based on which, the optimal BS location is analyzed to obtain a performance lower bound of the original problem. It is proved that with any given number of BSs, the minimum CRLB is exactly the same since the diversity gain exactly offsets the performance loss of each measurement. Then, the problem is proved to be a monotonic optimization (MO) problem and can be optimally solved by the Polyblock-based algorithm.
	\item For multi-vehicle cases, the relationship between the localization performance and the number of time slots, BSs, and targets is analyzed. Then, we propose a heuristic algorithm to optimize the target association and time allocation separately according to the derived insights for the single-target cases. 
	\item Finally, our simulation results verify the effectiveness of multi-target localization by leveraging IRSs and validate the superiority of the proposed schemes over benchmark schemes. Results show that deploying IRS on vehicles improves the resolution capability of multi-vehicle positioning and reduces the BS number requirements. 
\end{itemize}
\textit{Notations}: $\|{\bm{x}} \|$ represents the Euclidean norm of a complex-valued vector ${\bm{x}}$. $[{\bm{x}}]_{n}$ denotes the $n$th element of ${\bm{x}}$. For a general matrix ${\bm{X}}$, ${\bm{X}}^T$, and $[{\bm{X}}]_{m,n}$ respectively denote its transpose and the element in the $m$th row and $n$th column. For a square matrix ${\bm{Y}}$, ${\rm{Tr}}({\bm{Y}})$ and ${\bm{Y}}^{-1}$ denote its trace and inverse, respectively. $\mathbb{E}_x\{ \cdot \}$ denotes statistical expectation over the distribution of $x$, and $\{\cdot\}$ represents a variable set. $\otimes$ denotes the Kronecker product. 

\begin{figure}[t]
	\centering
	\setlength{\abovecaptionskip}{0.cm}
	\subfigure[IRS-assisted localization Scenarios.]
	{	
		\label{figure1a}
		\includegraphics[width=5.2cm]{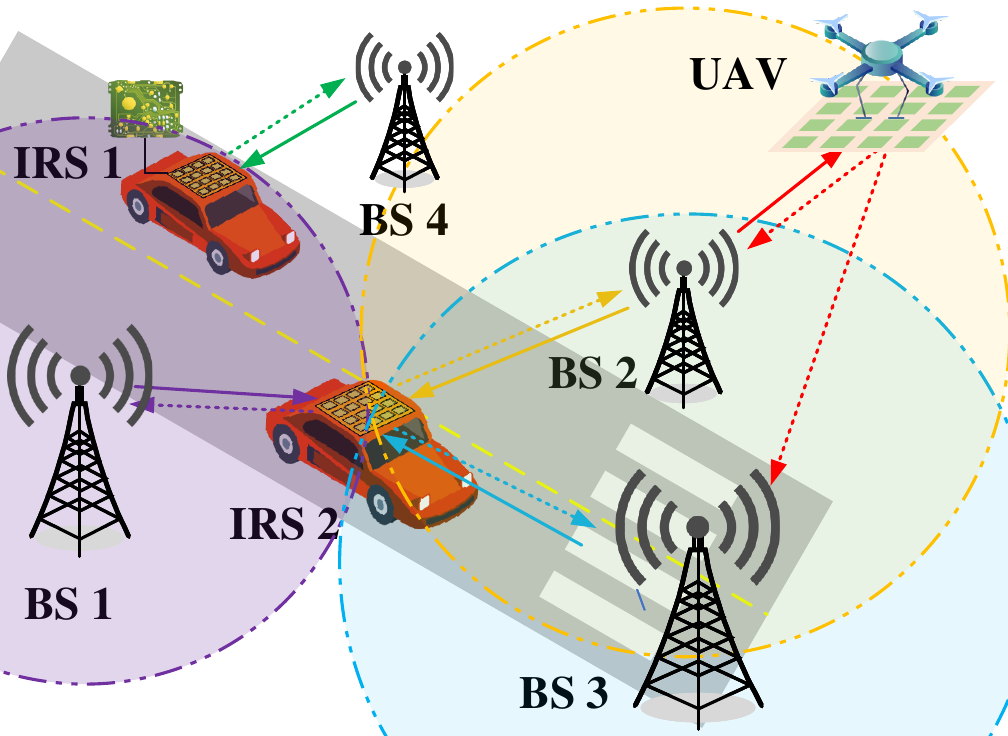}
	}
	\subfigure[Localization protocol.]
	{	
		\label{figure1b}
		\includegraphics[width=3cm]{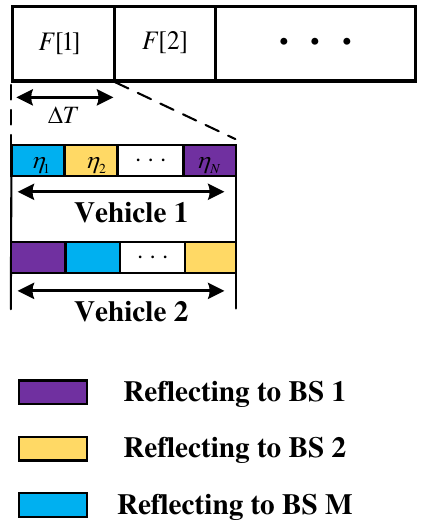}
	}	
	\caption{Scenarios and protocol for IRS-assisted localization.}
	\label{figure1}
\end{figure}
\section{System Model and Problem Formulation}
\label{AOAEstimationErrorSec2}

As shown in Fig. \ref{figure1}, $M$ cellular BSs with a single transmit antenna and a single receive
antenna are employed to provide localization services for $K$ vehicles in addition to communication services, i.e., the hardware of the communication receiver is actually reused for radar sensing. The targets can be vehicles, UAVs, robots, etc., and the focus of this work is primarily centered on vehicles, which are adopted as the primary exemplification. The proposed scheme can also be applied to multi-antenna cellular networks for providing localization services almost without affecting their communication performance since each BS can only use one transmit antenna for sensing. Here, we can equivalently estimate the IRS location as the target location due to its constant relative position, and the target and IRS are interchangeable according to context in the following discussion.\footnote{In practice, IRS can be easily attaches to diverse targets, such as vehicles and drones. It empowers targets to efficiently steer echo signals toward the desired direction for significantly improved positioning accuracy.} The BSs are indexed by $m \in {\cal{M}} = \{1, \cdots, M\}$, and the location of BS $m$ is denoted by $(x_m^r, y_m^r, H_{BS})$. The vehicles (representing targets in the following discussion) are indexed by $k \in {\cal{K}} = \{1, \cdots, K\}$, where an IRS is deployed on the surface of each vehicle for actively controlling echo signals and improving localization accuracy.\footnote{This positioning method solely relies on IRS interaction with the sensing transceivers and doesn't require any involvement from the target, thereby illustrating its robust scalability.} 
The uniform planar array (UPA) with half-wavelength antenna spacing is adopted at the IRS mounted on each vehicle, and the number of IRS elements is $L = L_x \times L_y$, indexed by $l \in {\cal{L}} = \{ 1,\cdots,L\}$, where $L_x$ and $L_y$ denote the number of elements along the $x$- and $y$-axis, respectively.

\subsection{Proposed Localization Scheme}

To tackle the target association issues, we propose a novel localization scheme to actively control the direction of echo signals by designing the IRS phase shifts, and thus facilitate to establish the association between targets and the measured distances at the BSs, as shown in Fig.~\ref{figure1}.\footnote{The angle information is not considered in this work since it requires multiple antennas equipped on the BSs, and the corresponding implementation cost and power consumption are high.} It is practically assumed that the location of vehicles keeps constant within a sufficiently small time duration $\Delta T$ \cite{Jayaprakasam2017Robust}. In the proposed localization scheme, the duration $\Delta T$ is divided into several time slots, and the association between BSs and IRSs is optimized according to the prior locations of vehicles. As shown in Fig.~\ref{figure1}, according to the designed association, the IRS phase shifts are optimized to boost the reflecting signal power towards the direction of the associated BS while reducing the interference to unassociated BSs, so as to achieve ultra-reliable and high-precision positioning with fewer BSs. Specifically, the time duration $\Delta T$ is divided into $N$ time slots, and $N$ is optimized according to the numbers of targets and BSs, which will be discussed in Section \ref{MultiVehicle}. The time slots are indexed by $n \in {\cal{N}} = \{1, \cdots, N\}$ and the proportion of time slot $n$ in $\Delta T$ is recorded as $\eta_n$ with $\sum\nolimits_{n=1}^N \eta_n = 1$. In this case, the dwell time at the $n$th time slot is $\eta_n \Delta T$.\footnote{The transmission duration of the associated BS could be set a little longer than the dwell time on the IRS, thereby improving the robustness of synchronization.} It is worth noting that in the proposed localization scheme, the BSs do not need to interact with the target, but only send control commands to the IRS controller to adjust its phase shifts. Due to the adopted device-free positioning method, the proposed algorithm is extremely tolerant to synchronization errors between IRSs and BSs \cite{Liu2022Survey}. 

To establish a unique association between BSs and IRSs, we introduce the set of binary variables $\{b_{k,m,n}\}$, $k \in {\cal{K}}, m \in {\cal{M}}, n \in {\cal{N}}$, which indicates that BS $m$ transmits signal $s_{m}(t)$, and then IRS $k$ reflects incident signal $s_{m}(t)$ towards BS $m$ for positioning vehicle $k$ at time slot $n$ if $b_{k,m,n} = 1$; when $b_{k,m,n} = 0$, it denotes that BS $m$ cannot be associated with IRS $k$ but can still be associated with other IRSs. The established unique association can effectively reduce the complexity of distance estimation and target localization. To establish the unique association between measured distances and vehicles, at time slot $n$, the phase shifts of IRS $k$ are designed to enhance the reflected echo signal power for at most one BS, i.e,\footnote{It is observed that the performance of our adopted association method is very close to that of the mechanism where multiple BSs are associated with one IRS in the same time slot, while effectively avoiding signalling overhead due to the low complexity of phase shift design.}
\begin{equation}
	\sum\nolimits_{m = 1}^M b_{k,m,n} \le 1, \forall k \in {\cal{K}}, n \in {\cal{N}}.
\end{equation}
If $\sum\nolimits_{m = 1}^M b_{k,m,n} = 0$, IRS $k$ does not reflect signals to any BS and it works as an absorbing metasurface at the $n$th time slot \cite{gong2020toward}, which generally happens when $K > M$. Otherwise, IRS $k$ reflects the incident signals transmitted by all active BSs. Similarly, to avoid mutual interference between BSs and IRSs, at each time slot, each BS transmits signals to sense at most one target, i.e.,
\begin{equation}
	\sum\nolimits_{k = 1}^K b_{k,m,n} \le 1, \forall m \in {\cal{M}}, n \in {\cal{N}}.
\end{equation}
Moreover, to simplify the practical implementation, we assume that one target and one BS are associated in at most one time slot as follows:
\begin{equation}
	\sum\nolimits_{n = 1}^N b_{k,m,n} \le 1, \forall k \in {\cal{K}}, m \in {\cal{M}}.
\end{equation}
By adopting the above constraints, the distance measured at the BSs can be associated with a specific target, thereby effectively improving the localization accuracy and avoiding the complexity of target association. In practice, the BS-IRS association is optimized at the BSs according to the prior location information, based on which the IRS phase shifts are designed to enhance the echo signal towards the associated BS while reducing the interference to unassociated BSs. Then, the designed phase shifts are transmitted to the IRS controller before each positioning operation.

Without loss of generality, the prior location of vehicle $k$, denoted by $(\hat x_k^t, \hat y_k^t)$, can be obtained based on its onboard sensors such as global positioning system (GPS) \cite{Liu2023CLosedRadarComm} or the state estimation in previous observations. The accurate 2D location of vehicle $k$, i.e., $(x_k^t, y_k^t)$, is assumed to be uniformly distributed within a circle with the radius $r_e$ and center point $(\tilde x_k^t, \tilde y_k^t)$ due to vehicle mobility or sensor measurement errors, where $(\tilde x_k^t, \tilde y_k^t)$ represents the prior location of IRS $k$. The actual distance between IRS $k$ and BS $m$ is expressed by 
\begin{equation}
	d_{k,m} = \sqrt{(x_m^r -  x_k^t)^2 + (y_m^r -  y_k^t)^2 + (H_{BS} - H_{\rm{IRS}})^2},
\end{equation}
where $H_{\rm{IRS}}$ represents the height of the IRS, and it is practically assumed to be known due to the fixed IRS deployment position relative to the vehicle. For notational simplicity, vehicles are assumed to drive along a straight road that is parallel to the $x$-axis. At the $n$th time slot, the reflection-coefficient matrix of IRS $k$ is given by ${\bm{\Theta}}_{k,n} = {\rm{diag}}(e^{j \theta_{k,1,n}}, ... , e^{j \theta_{k,L,n}})$, where $\theta_{k,l,n} \in [0, 2\pi)$ denotes the reflecting phase shifts of the $l$th element of IRS $k$.  

\subsection{Radar Measurement Model}
\label{RadarMeasurement}
The low-pass equivalent of the radar signal transmitted from BS $m$ is denoted by $s_{m}(t)$,\footnote{The transmitted signals may also be communication signals, which leads to an interesting future work for integrated sensing and communication (ISAC) applications \cite{meng2022throughput}.} and $\frac{1}{W} \int_{W \cdot \Delta t}\left|s_m(t)\right|^2 d t=1$, where $W$ is the number of symbols during signal processing interval and $\Delta t$ is the time of one symbol. Following the assumption in \cite{godrich2010target, sadeghi2021target}, the transmitted signals are approximately orthogonal for any time delay $\tau$ of interest, i.e.,\footnote{Information encoding methods can be adopted to ensure the orthogonality of signals transmitted by different BSs, such as code division multiple access (CDMA).}
\begin{equation}\label{FilterInterference}
	\frac{1}{W}\int_{W \cdot \Delta t} s_m(t) s_{m^{\prime}}^*(t-\tau) d t \approx\left\{\begin{array}{lll}
		1 & \text { if } m=m^{\prime} \\
		0 & \text { if } m \neq m^{\prime}
	\end{array}\right.,
\end{equation}
where $(\cdot)^*$ denotes the conjugate operator. The vehicle $k$’s direction relative to BS $m$ is defined by $\{\varphi_{k,m}, \phi_{k,m}\}$, where $\varphi_{k,m}$ and $\phi_{k,m}$ respectively denote the azimuth and elevation angles of the geometric path connecting IRS $k$ and BS $m$. The channel power gain between IRS $k$ and BS $m$ can be given by $\beta^G_{k,m} = \beta_0 d_{k,m}^{-2}$, where $\beta_0$ is the channel power gain at the reference distance 1 m (meter). $\bm{h}^{\mathrm{DL}}_{k,m}  \in \mathbb{C}^{1 \times L}$ and $\bm{h}^{\mathrm{UL}}_{k,m} \in \mathbb{C}^{L \times 1}$ are respectively the downlink and uplink channel vectors between BS $m$ and IRS $k$, given by
\begin{equation}\label{SteeringVector}
	\bm{h}^{\mathrm{DL}}_{k,m}=\sqrt{\beta^G_{k,m}} \bm{a}_{\mathrm{IRS}}\left(\varphi_{k,m}, \phi_{k,m}\right),
\end{equation}
\begin{equation}
	\bm{h}^{\mathrm{UL}}_{k,m}=\sqrt{\beta^G_{k,m}} \bm{a}_{\mathrm{IRS}}^{T}\left( \varphi_{k,m}, \phi_{k,m}\right), 
\end{equation}
where $\bm{a}_{\mathrm{IRS}} = \left[1, \cdots, e^{ {-j \pi\left(L_{x}-1\right) \Phi_{k,m} }}\right]^T  \otimes\left[1, \cdots, e^{ { -j  \pi\left(L_{y}-1\right) \Omega_{k,m} }}\right]^T$, $\Phi_{k,m} = \sin (\phi_k) \cos (\varphi_k)$, and $\Omega_{k,m} = \sin (\phi_k) \sin (\varphi_k)$. Here, $\Phi_m$ and $\Omega_m$ respectively denote the spatial frequency angles from BS $m$ to the IRS along the $x$-axis and $y$-axis.

Let $\mu_{k,m}$ and $\tau_{k,m}$ respectively denote Doppler frequency and the round-trip delay of echo signals transmitted/received by BS $m$. The transmission delay from BS $m'$ to IRS $k$ and then to BS $m$ is denoted by $\tau_{k,m,m'}$, and that from BS $m'$ to BS $m$ is denoted by $\tau_{m,m'}$. At the $n$th time slot, the echo signals received at BS $m$ can be expressed as\footnote{The reflected echo from the vehicle surface (non-IRS parts) is practically much weaker than the controllable echo signals of the IRSs, and the sensing performance is generally improved with the exploiting the echo signal reflected from the vehicle surfaces.}
\begin{align}
	&{{r}}_{m,n}(t) =  \underbrace{b_{k,m,n} {{g}}_{k,m,m,n} \sqrt{P_{\mathrm{A}}} s_{m}(t - \tau_{k,m})}_{\text{Reflected from associated target}} \nonumber \\
	& + \underbrace {\sum\nolimits_{m' \ne m}^M \sum\nolimits_{k = 1}^K b_{k,m',n} {{g}}_{k,m,m',n} \sqrt{P_{\mathrm{A}}} s_{m'}(t - \tau_{k,m,m'})}_{{\text{Echo signals transmitted by other BSs}}} \nonumber \\
	& + \underbrace {\sum\nolimits_{m' \ne m}^M \sum\nolimits_{k = 1}^K b_{k,m',n}  {{h}}_{m,m',n} \sqrt{P_{\mathrm{A}}} s_{m'}(t - \tau_{m,m'})}_{{\text{Interference signals directly transmitted by other BSs}}} \nonumber \\
	& + \underbrace  {b_{k,m,n}\sum\nolimits_{k' \ne k}^K  {{g}}_{k',m,m,n} \sqrt{P_{\mathrm{A}}} s_{m}(t - \tau_{k',m})}_{{\text{Interference reflected from unassociated IRSs}}} + {{z}}_m(t),
\end{align}
where  ${{g}}_{k,m,m',n} = e^{j2\pi \mu_{k,m} t} {\bm{h}}^{\mathrm{UL}}_{k,m} {\bm{\Theta}}_{k,n} {\bm{h}}^{\mathrm{DL}}_{k,m'}$, ${{h}}_{m,m',n}$ denotes the channel between BS $m$ and BS $m'$, $P_{\mathrm{A}}$ is the transmit power, ${{z}}_m(t) \in \mathcal{C} \mathcal{N}(0, \sigma_s^2 )$ represents the additive disturbance, and $\sigma_s^2$ is the noise power at the receive antennas of the BS. Note that the interference from other unassociated BSs can be removed by the filtering operations due to the orthogonality of signals in \eqref{FilterInterference}. The BS correlates the received signal with the associated normalized detection signal, and the  signal-to-interference-plus-noise ratio (SINR) of the received echos reflected from IRS $k$ is expressed in expectation form due to uncertain prior location of vehicles, i.e.,
\begin{align}\label{SensingReceivedPower}
	\gamma^{\mathrm{S}}_{k,m} =& \sum\nolimits_{n = 1}^N b_{k,m,n}  \frac{{\eta_{n} \Delta T}}{ \Delta t} \nonumber \\
	& \times \frac{P_{\mathrm{A}} \mathbb{E}_{\varphi_{k,m}, \phi_{k,m}} \left[ \left|  {{g}}_{k,m,m,n} \right|^{2} \right]}{ P_{\mathrm{A}} \sum\nolimits_{k' \ne k}^K   \mathbb{E}_{\varphi_{k',m}, \phi_{k',m}} \left[ \left|  {{g}}_{k',m,m,n}\right|^{2} \right] + \sigma_s^2}.
\end{align}
In (\ref{SensingReceivedPower}), $\frac{\Delta T}{\Delta t}$ represents the number of symbols during the time $\Delta T$. In (\ref{SensingReceivedPower}), the interference mainly arises from the echo signals reflected from unassociated IRSs. The interference relationship between IRSs and BSs is further analyzed in Section \ref{MultiVehicle}.

Based on the above analysis, by adopting a certain association $\{b_{k,m,n}\}$, the vector parameter ${\bm{p}}_k = [x^t_k, y^t_k]$ can be estimated based on the measured distance, e.g., least squares method and maximum likelihood estimation (MLE) method \cite{myung2003tutorial}. In this work, the MLE method is adopted to obtain the vehicle location. The estimated distance from BS $m$ to IRS $k$ is expressed as
\begin{equation}
	\tilde d_{k,m} = d_{k,m} + z_{k,m},
\end{equation}
\begin{figure*}[t]
	\begin{equation}\label{DefinitionCRLB}
		\begin{aligned}
			{\rm{CRLB}}_k \!=\! {\rm{tr}}\left( \mathcal{I}^{-1}(\boldsymbol{p}_k) \right) 
			\!= \! \frac{{\sum\nolimits_{m = 1}^M {\frac{{{{ {\cos^2 {\phi _{k,m}}} }}}}{{\sigma _{k,m}^2}}} }}{{\sum\nolimits_{m = 1}^M {\frac{{{{ {\cos^2 {\varphi _{k,m}}\cos^2 {\phi _{k,m}}} }}}}{{\sigma _{k,m}^2}}\sum\nolimits_{m = 1}^M {\frac{{{{ {\sin^2 {\varphi _{k,m}}\cos^2 {\phi _{k,m}}} }}}}{{\sigma _{k,m}^2}}} \! - \! {{\left( {\sum\nolimits_{m = 1}^M {\frac{{\cos {\varphi _{k,m}}\cos {\phi _{k,m}}\sin {\varphi _{k,m}}\cos {\phi _{k,m}}}}{{\sigma _{k,m}^2}}} } \right)}^2}} }}.
		\end{aligned}
	\end{equation}
	\hrulefill
\end{figure*}
where the error of estimated distance $z_{k,m}$ is generally inversely proportional to the SINR of echo signals at the BS receiver \cite{zhang2023fast, dai2022composed}, i.e., 
\begin{equation}\label{ReceivedPowerEquation}
	\sigma_{{k,m}}^{2} \propto ({{\gamma^{\mathrm{S}}_{k,m} }})^{-1}, 
\end{equation}  
where ${\gamma^{\mathrm{S}}_{k,m}}$ is the SINR at the receive antenna of BS $m$ after match-filtering. Given a vector parameter ${\bm{p}}_k = [x^t_k, y^t_k]$, the unbiased estimate satisfies the following inequality
\begin{equation}
	E_{\hat{\boldsymbol{p}}_k}\left\{(\hat{\boldsymbol{p}}_k-\boldsymbol{p}_k)(\hat{\boldsymbol{p}}_k-\boldsymbol{p}_k)^T\right\} \geq \mathcal{I}^{-1}(\boldsymbol{p}_k),
\end{equation}
where $\mathcal{I}(\boldsymbol{p}_k)$ is the Fisher Information matrix (FIM), and it can be obtained by utilizing the chain rule based on the measurement ${\bm{d}}_k = [d_{k,1},\cdots,d_{k,M}]$, when the measurement noise is Gaussian and the distance measurement covariance matrix $\boldsymbol{\Sigma}_k$ is not dependent on the parameter ${\bm{p}}_k$ \cite{torrieri1984statistical}, i.e., 
\begin{equation}\label{FIMequation}
	\mathcal{I}(\boldsymbol{p}_k)=\left(\nabla_{\boldsymbol{p}_k} \boldsymbol{d}(\boldsymbol{p}_k)\right)^T {\boldsymbol{\Sigma}}_k^{-1}\left(\nabla_{\boldsymbol{p}_k} \boldsymbol{d}(\boldsymbol{p}_k)\right),
\end{equation}
where the distance measurement covariance matrix of $M$ BSs is denoted by
\begin{equation}\label{DistnaceCovariance}
	\boldsymbol{\Sigma}_k=\left[\begin{array}{ccc}
		\sigma_{k,1}^2 & \cdots & 0 \\
		\vdots & \ddots & \vdots \\
		0 & \cdots & \sigma_{k,M}^2
	\end{array}\right]_{M \times M}.
\end{equation}
In (\ref{FIMequation}), the Jacobian matrix $\nabla_{\boldsymbol{p}_k} \boldsymbol{d}(\boldsymbol{p}_k)$ is $\nabla_{\boldsymbol{p}_k} \boldsymbol{d}(\boldsymbol{p}_k) = \frac{\partial {\bm{d}}_k}{\partial {\bm{p}_k}}$, i.e., 
\begin{equation}\label{JacobianMatrix}
	\nabla_{\boldsymbol{p}_k} \boldsymbol{d}(\boldsymbol{p}_k) = \left[ {\begin{array}{*{20}{c}}
			{\cos {\varphi _{k,1}}\cos {\phi _{k,1}}}&{\sin {\varphi _{k,1}}\cos {\phi _{k,1}}}\\
			\vdots & \vdots \\
			{\cos {\varphi _{k,M}}\cos {\phi _{k,M}}}&{\sin {\varphi _{k,M}}\cos {\phi _{k,M}}}
	\end{array}} \right].
\end{equation}
By plugging (\ref{DistnaceCovariance}) and (\ref{JacobianMatrix}) into (\ref{FIMequation}), the CRLB of target $k$'s location estimation error can be given by (\ref{DefinitionCRLB}), which is shown at the top of the page.

\subsection{Problem Formulation}
In this work, we aim to minimize the maximum CRLB of the location estimation error by jointly optimizing the BS-IRS association, the phase shift matrices of IRSs, and the time allocation. The problem is formulated as follows.
\begin{alignat}{2}
	\label{P1}
	(\rm{P1}): \quad & \begin{array}{*{20}{c}}
		\mathop {\min }\limits_{\{{\bm{\Theta}}_{k,n}\}, {\bm{\eta}}, \{b_{k,m,n}\}, N} \quad   \mathop {\max }\limits_k \  {\rm{CRLB}}_{k}
	\end{array} & \\ 
	\mbox{s.t.}\quad
	& \theta_{k,l,n} \in [0, 2 \pi), \forall k \in {\cal{K}}, l \in {\cal{L}}, n \in {\cal{N}}, \tag{\ref{P1}a} \\
	&  \sum\nolimits_{n = 1}^N \eta_{n} = 1, \tag{\ref{P1}b} \\
	& \eta_{n} \in [0, 1], \forall n \in {\cal{N}}, & \tag{\ref{P1}c} \\
	& \sum\nolimits_{k = 1}^K b_{k,m,n} \le 1, \forall m \in {\cal{M}}, n \in {\cal{N}}, & \tag{\ref{P1}d} \\
	& \sum\nolimits_{m = 1}^M b_{k,m,n} \le 1, \forall k \in {\cal{K}}, n \in {\cal{N}},  & \tag{\ref{P1}e} \\
	& \sum\nolimits_{n = 1}^N b_{k,m,n} \le 1, \forall k \in {\cal{K}}, m \in {\cal{M}}, & \tag{\ref{P1}f} \\
	& b_{k,m,n} \in \{0,1\}, \forall k \in {\cal{K}}, m \in {\cal{M}}, n \in {\cal{N}}, & \tag{\ref{P1}g} \\
	& N \in  \mathbb{N}^+ , & \tag{\ref{P1}e}
\end{alignat}
where ${\bm{\eta}} = [\eta_{1},\cdots,\eta_N]$. Solving (P1) optimally is non-trivial due to the non-convex objective function and the closely coupled variables. To tackle this issue, we first propose a phase shift design method based on beam flattening, and analyze the number of optimal associated BSs to simplify the formulated problem in Section \ref{SingleVehicle}. Then, the relationship between the numbers of time slots, BSs, and targets for the multi-vehicle case is analyzed in Section \ref{MultiVehicle}.

\section{Single-Vehicle Localization}
\label{SingleVehicle}

In this section, we consider the single-vehicle setup, i.e., $K=1$, to draw useful insights into the IRS phase shift and time allocation design. If more than $M$ time slots are employed, some BSs may be assigned multiple time slots, in this case, merging multiple time slots associated with the same BS does not compromise the sensing performance. Therefore, without loss of generality, the number of time slots is set to $N = M$. Then, variables $\{b_{m,n}\}$ can be ignored since the association order does not affect the CRLB value and the BS-IRS association can be recovered according to the time allocation results. For simplicity, the IRS is associated with the BSs in the order of the BS index, and in this case, the time allocation variable can be denoted by $\{\eta_m\}_{m=1}^M$. Then, (P1) is simplified to (by dropping the IRS index) minimizing the CRLB of the localization error, i.e., 
\begin{alignat}{2}
	\label{P2}
	(\rm{P2}): \quad & \begin{array}{*{20}{c}}
		\mathop {\min }\limits_{\{{\bm{\Theta}}_{m}\}, {\bm{\eta}}} \quad    {\rm{CRLB}}
	\end{array} & \\ 
	\mbox{s.t.}\quad
	& \theta_{l,m} \in [0, 2 \pi), \forall l \in {\cal{L}}, m \in {\cal{M}}, \tag{\ref{P2}a} \\
	&  \sum\nolimits_{m = 1}^M \eta_{m} = 1, \tag{\ref{P2}b} \\
	& \eta_{m} \in [0, 1], \forall m \in {\cal{M}}. & \tag{\ref{P2}c} 
\end{alignat}

In the following, we first present a phase shift design method for robust localization, and derive the simplified expression of the CRLB with respect to time allocation ${\bm{\eta}}$. Then, under the case that the location of BSs can be jointly optimized, we derive the optimal location of the associated BSs (i.e., $\eta_m > 0$) and a lower bound of CRLB. According to the derived conclusions, we propose an MO-based algorithm to optimally solve the time allocation problem for the case where the location of BSs is fixed.

\subsection{Phase Shift Design for Robust Localization}
In the considered system, the phase shift design to achieve efficient and robust localization is challenging since the accurate location $(x^t, y^t)$ is uncertain and the IRS phase shifts should be obtained and transmitted to the IRS controller with ultra-low latency. To this end, we propose an efficient solution based on the beam flattening technique \cite{Lu2021AerialSurface}. Specifically, due to the uncertain location of the vehicle, spatial resolution angles from BS $m$ to the IRS (i.e., $\Phi_{m}$ and $\Omega_{m}$) lie within two angular spans, denoted by $\left[\bar \Phi_m, \underline \Phi_m\right]$ and $\left[\bar \Omega_m, \underline \Omega_m\right]$, where $\bar \Phi_m = \mathop {\max }\limits_{{{x}^t_m,{y}^t_m}} \Phi_m$, $\underline \Phi_m = \mathop {\min }\limits_{{{x}^t_m,{y}^t_m}} \Phi_m$, $\bar \Omega_m = \mathop {\max }\limits_{{{x}^t_m,{y}^t_m}} \Omega_m$, and $\underline \Omega_m = \mathop {\min }\limits_{{{x}^t_m,{y}^t_m}} \Omega_m$. By adopting beam flattening \cite{Lu2021AerialSurface}, when the IRS is associated with BS $m$, the IRS elements of each row along the $x$-axis are divided into $Q^x_m$ sub-arrays with $L^s_x = L_x/Q^x_m$ elements in each sub-array. Similarly, the IRS elements of each column along the $y$-axis are divided into $Q^y_m$ sub-arrays with $L^s_y = L_y/Q^y_m$ elements in each sub-array, where $L_x/Q^x_m$ and $L_y/Q^y_m$ are assumed to be integers for ease of analysis. As a result, the IRS is divided into $Q^x_m \times Q^y_m$ parts to ensure that BS $m$ can receive the echo signal reflected from the IRS, where $Q^x_m = \sqrt{ { L_x \left|\bar \Phi_m - \underline \Phi_m\right|}/{2}}$ and $Q^y_m = \sqrt{ { L_y \left|\bar \Omega_m - \underline \Omega_m\right|}/{2}}$. 
Here, the IRS phase shifts can be designed based on the beam flattening technique proposed in \cite{Lu2021AerialSurface} and the BSs only need to transmit the angular span information ($\bar \Phi_m$ and $\underline \Phi_m$) to the IRS controller for phase shift design. Then, the effective beamforming gain of one subarray is a function of the spatial frequency distance ($\Delta_{\Phi}$ and $\Delta_{\Omega}$) to the beam center, denoted by $P(\Delta_{\Phi}, \Delta_{\Omega})=\frac{\sin \left(\frac{\pi L_x^s  \Delta_{\Phi}}{2} \right)}{\sin \left(\frac{\pi  \Delta_{\Phi}}{2}\right)} \frac{\sin \left(\frac{\pi L_y^s \Delta_{\Omega}}{2} \right)}{\sin \left(\frac{\pi  \Delta_{\Omega}}{2}\right)}.$
Specifically, $P(\Delta_{\Phi}, \Delta_{\Omega})$ has a peak at $\Delta_{\Phi}=0$ and $\Delta_{\Omega} = 0$, and nulls at $\Delta_{\Phi}= \pm \frac{2k}{L_x^s}, k=1, \cdots, L_x^s-1$ and $\Delta_{\Omega}= \pm \frac{2j}{L_y^s}, j=1, \cdots, L_y^s-1$. Then, the coverage beamwidth can be deemed as $\frac{2}{L_x^s}$ and $\frac{2}{L_y^s}$ along the $x$-axis and $y$-axis, respectively. To obtain a flattened beam, the adjacent spatial frequency shift is separated by the spatial frequency resolution of the subarray, i.e., $\frac{2}{L_x^s}$ and $\frac{2}{L_y^s}$. Hence, the $m$th and $n$th sub-array beam's directions along the $x$ and $y$ axes can be expressed as $\bar{\Phi}_m^{l_x} = \bar{\Phi}_m^1 + \frac{2(l_x-1)}{L_x^s}$ for $l_x=1, \cdots, Q^x_m$, and $\bar{\Omega}_m^{l_x} = \bar{\Omega}_m^1 + \frac{2(l_y-1)}{L_y^s}$ for $l_y=1, \cdots, Q^y_m$ respectively. 
By adopting beam flattening, even if the position of the vehicle is uncertain, the BS can always receive the echo signal reflected from the IRS, as shown in  Fig.~\ref{figure3}. Then, the 
signal-to-noise ratio (SNR) of echo signals received at BS $m$ can be recast in approximate form as
\begin{equation}
	\gamma_{m}^S \approx \eta_m \bar \gamma_{m}^S,
\end{equation} 
where $Q_m = Q_m^x Q_m^y$ and $\bar \gamma_{m}^S = \frac{\beta^2_0 \Delta T L^2}{ Q^2_m \Delta t  d_{m}^4 \sigma_s^2}$. If the given error bound of the target's location is less than a range (i.e., less than the beam width), $Q_m = 1$, i.e., all the beam is aligned towards the prior location. Obviously, the larger the angular range, the more dispersed the echo power and the correspondingly lower distance measurement accuracy. 

\begin{figure}[t]
	\centering
	\includegraphics[width=8cm]{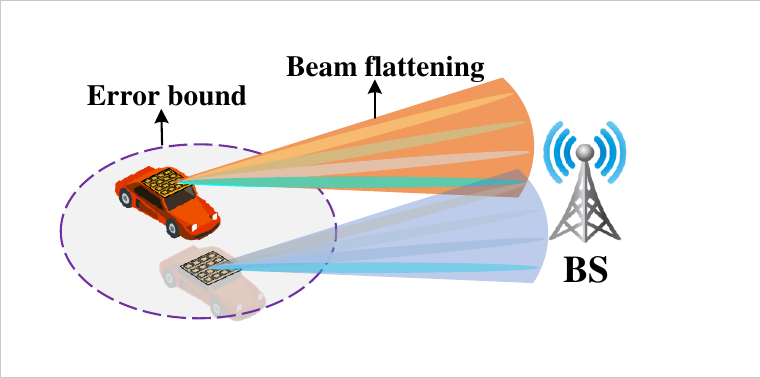}
	\caption{Illustration of beam flattening for robust localization.}
	\label{figure3}
\end{figure} 

According to (\ref{ReceivedPowerEquation}), we have $\sigma_{{m}}^{2} = \frac{C_0}{{{\gamma^{\mathrm{S}}_{m} }}}$, where $C_0$  is the variance parameter of the estimation method \cite{Liu2022Survey}. Let $\tilde \gamma _{m} = \bar \gamma_{m} \cos^2\phi_{m}$, we can equivalently simplify the CRLB expression as
\begin{equation}\label{SimplifiedCLEB}
	\begin{aligned}	
		{\rm{CRLB}} =	\frac{C_0{\sum\nolimits_{m = 1}^M {{\eta _m}} {{{\tilde \gamma }_{m}}} }}{{\sum\nolimits_{j = 1}^{M-1} {{{\eta _j}} {{\tilde \gamma }_{j}} \left( {\sum\nolimits_{i = j+1  }^M {{{{{\eta _i}} {{\tilde \gamma }_{i}}{\sin^2 (\varphi_i - \varphi_j)}}}} } \right)} }}.
	\end{aligned}
\end{equation}

The simplified expression of CRLB in (\ref{SimplifiedCLEB}) shows that the sensing accuracy is only related to the azimuth difference between any two associated BSs to the target, i.e., $\varphi _j - \varphi _i$, instead of the absolute azimuth angle. Nonetheless, it is still challenging the optimally solve (P2) since the CRLB is non-convex with respect to time allocation ${\eta}$. To tackle this issue, we first analyze the minimum number of associated BSs with flexible BS locations. Then, we prove (P2) is a monotonic optimization (MO) problem with respect to time allocation, based on which, (P2) can be optimally solved by Polyblock-based algorithm \cite{tuy2000monotonic}.

\subsection{Lower Bound under Flexible BS Deployment}
\label{FlexibleBS}
In this subsection, we derive the optimal location of associated BSs and the optimal time allocation results. In practice, the distance from BS $m$ to the target satisfies $d_{m} \ge \underline d$, where $\underline d$ is the minimum distance with the consideration of safety.

\begin{thm}\label{LowerBoundCRLB}
	For any given $M \ge 3$, if the location of BSs can be jointly optimized and $Q_m = 1, \forall m$, the minimum CRLB in (P2) can be given by
	\begin{equation}\label{AccuracyBound}
			{\rm{CRLB}}^*  = \frac{4 C_0 \Delta t  \underline d^6 \sigma_s^2}{ \Delta T \beta^2_0 L^2  \left(\underline d^2 - (H_{BS} - H_{IRS})^2\right)},
	\end{equation}
	where the equality holds when $\eta^*_m = 1/M$ and $\varphi_m^* = \frac{2(m-1)\pi}{M}$,  $\phi_m^* = \arcsin{\frac{|H_{BS} - H_{IRS}|}{\underline d}}, m = 1, \cdots, M$.
\end{thm}
\begin{proof}
	Please refer to Appendix A. 
\end{proof}

Lemma \ref{LowerBoundCRLB} shows that with flexible BS deployment, the minimum CRLB is inversely proportional to the equivalent SNR $\tilde \gamma^S_m$. Moreover, it can be found that the minimum CRLB with different numbers of BSs is exactly equal to each other, i.e., the optimal localization performance of the proposed localization system is independent of the number of BSs. This is in sharp contrast to the traditional distributed radar sensing systems, in which the corresponding localization accuracy is generally improved with the increase of the transmitters' number. The main reason is that with the given time interval $\Delta T$, if the time is allocated to more BSs and the average time resource allocated to each BS will enviably decrease due to the limited total dwell time $\Delta T$, and the localization performance improvement brought by sensing direction diversity cannot compensate for the performance loss caused by the time reduction for each distance measurement.

When there is prior knowledge about the region of the target, two receivers ($M = 2$) would suffice for target localization. It is not difficult to verify that if $M = 1$, ${\rm{CRLB}} \to \infty$. Hence, as shown in Lemma \ref{LowerBoundCRLB}, the achievable minimum CRLB values under different numbers of BSs are all equal to ${\rm{CRLB}}^*$ in (\ref{AccuracyBound}), thus the minimum number of associated BSs is two. 

\begin{Pro}\label{2BS_optimal}
	When $M = 2$, the optimal time allocation $\eta^*_1 = \frac{\sqrt{\tilde \gamma_2}}{\sqrt{\tilde \gamma_1} + \sqrt{\tilde \gamma_2}}$ and $\eta^*_2 = \frac{\sqrt{\tilde \gamma_1}}{\sqrt{\tilde \gamma_1} + \sqrt{\tilde \gamma_2}}$.
\end{Pro}
\begin{proof}
	When $M = 2$, corresponding minimum CRLB can be given by
	\begin{equation}\label{2BSCases}
		{\rm{CRLB}} = \frac{C_0}{\sin^2(\varphi_1 - \varphi_2)}\left(\frac{1}{\eta_1 \tilde \gamma_1} + \frac{1}{\eta_2 \tilde \gamma_2}\right).
	\end{equation}
	At the optimal solution, $\frac{{\rm{CRLB}}}{\eta_1} = \frac{{\rm{CRLB}}}{\eta_2}$. Then, it can be readily proved that the optimal time allocation $\eta^*_1 = \frac{\sqrt{\tilde \gamma_2}}{\sqrt{\tilde \gamma_1} + \sqrt{\tilde \gamma_2}}$ and $\eta^*_2 = \frac{\sqrt{\tilde \gamma_1}}{\sqrt{\tilde \gamma_1} + \sqrt{\tilde \gamma_2}}$.
\end{proof}

In Proposition \ref{2BS_optimal}, ${{{{\left( {\cos {\varphi _1}\sin {\varphi _1} - \cos {\varphi _2}\sin {\varphi _2}} \right)}^2}}} \in [0, 1]$. Thus, if $\varphi_1 = \varphi_2$, ${\rm{CRLB}} \to \infty$. CRLB is maximized if $\varphi_2 - \varphi_1 = \frac{\pi}{2}$. Based on the above conclusions, the higher the deployment density of BSs, the higher the probability of achieving the optimal positioning performance in Lemma \ref{LowerBoundCRLB}. 

\subsection{Optimal Solution to (P2) under Fixed BS Deployment}
\label{OptimalFixDeployment}
According to the analysis in Section \ref{FlexibleBS}, if the location of BSs satisfies the condition in Lemma \ref{LowerBoundCRLB} or Proposition \ref{2BS_optimal}, the optimal time allocation solution with flexible BS location is also the optimal solution to (P2). In the following, we will further present an analysis of fixed BS deployment when $M \ge 3$.

\begin{Pro}\label{MonoProof}
The ${\rm{CRLB}}$ value decreases monotonically as $\eta_m$ increases.
\end{Pro}
\begin{proof}
	Please refer to Appendix B.
\end{proof}

Intuitively, the increase of the dwell time on any BS leads to the improvement of measurement accuracy, thereby improving the positioning accuracy of the target. Proposition \ref{MonoProof} implies that (P2) is an MO problem with respect to $\eta_m$, and thus (P2) can be optimally solved based on the framework of the Polyblock-based algorithm \cite{tuy2000monotonic, zhang2013monotonic}. Specifically, during the $r$th iteration, the vector ${\bm{\eta}}^{(r)} =
[\eta_1^{(r)},\cdots,\eta^{(r)}_M]$ corresponding to the minimum CRLB is selected, and then its projection point is calculated as $\eta^{(r)}_m = \frac{\eta^{(r-1)}_m}{\sum\nolimits_{m=1}^M \eta^{(r-1)}_m}$.  After each iteration, a smaller Polyblock set can be constructed by replacing the vertices ${\bm{\eta}}^{(r)}$ with the newly generated vectors. The algorithm details are omitted for brevity, more details refer to \cite{zhang2013monotonic}. 

\begin{remark}
	This algorithm can also be extended to the case in which at least three BSs are required to be associated with, i.e., the number of non-zero elements in the set $\{ \eta_1, \cdots, \eta_M\}$ must be no less than three, and this problem is denoted by problem (P2.1). Specifically, the minimum time ratio for each BS to achieve effective sensing is denoted by $t_m$. If the obtained number of associated BSs for the above algorithm is two, with optimal time allocation $\eta_i^*$ and $\eta_j^*$, $i,j \in {\cal{M}}$, we can construct an optimal solution to (P2.1) as follows. The measurement from another BS with the minimum resulting CRLB is added, denoted by BS $q$, and then the dwell time of the original solution to problem (P2) is decreased in scale. As a result, the optimal solution to (P2.1) can be obtained by allocating time ratio $t_m$ to BS $q$, and the optimal time ratio for BSs $i$ and $j$ are ${\eta_i^*(1-t_m)}$ and ${\eta_j^*(1-t_m)}$, respectively. 
\end{remark} 

For ease of discussion, we shall take ${\rm{CRLB}}^*_m$ for the optimal solution where the maximum number of BSs allowed to be associated is $m$. Next, it is illustrated that each target can achieve effective positioning even associating with only a few BSs.

\begin{Pro}\label{SensingAccuracy}
	If $\tilde \gamma_{m} = \tilde \gamma_{m'}$, $\forall m, m' \in {\cal{M}}$, it follows that 
	\begin{equation}\label{DeriveConclusion}
		{{\rm{CRLB}}^*_2} \le 2{ \lim_{M \rightarrow \infty} {\rm{CRLB}}^*_M}.
	\end{equation}
\end{Pro}
\begin{proof}
	Furthermore, when there are only two BSs that can be associated with the IRS, we have ${\rm{CRLB}}^*_2 = \frac{4C_0}{\tilde \gamma_m \mathop {\max }\limits_{i \ne j, i, j \in {\cal{M}}} \sin^2(\varphi_i - \varphi_j) }$. When $M \rightarrow \infty$, we have 
	\begin{align}\label{InfityPerformance}
		&\lim_{M \rightarrow \infty} {\rm{CRLB}}^*_M \nonumber \\
		\ge&  \frac{C_0}{\tilde \gamma_{m} \mathop {\max }\limits_{i \ne j, i, j \in {\cal{M}}} \sin^2(\varphi_i - \varphi_j) \lim_{M \rightarrow \infty} \sum\nolimits_{i = 1}^{M-1} \sum\nolimits_{j = i + 1}^M \eta_i \eta_{j}} \nonumber \\
	 \ge& \frac{C_0}{\tilde \gamma_{m} \mathop {\max }\limits_{i \ne j, i, j \in {\cal{M}}} \sin^2(\varphi_i - \varphi_j) } \lim_{M \rightarrow \infty} \frac{2M^2}{M^2 - M} \nonumber  \\
		\overset{(a)}{=} &  \frac{2C_0}{\tilde \gamma_{m} \mathop {\max }\limits_{i \ne j, i, j \in {\cal{M}}} \sin^2(\varphi_i - \varphi_j) } \triangleq \frac{1}{2}{{\rm{CRLB}}^*_2},
	\end{align} 
where ($a$) holds due to $\lim_{M \rightarrow \infty} \frac{2M^2}{M^2 - M} = 2$. Combining the above analysis, we have (\ref{DeriveConclusion}) and complete the proof.
\end{proof}

Proposition \ref{SensingAccuracy} illustrates that ignoring the error difference of different BSs' measurements, when the number of BSs is large enough, the optimal CRLB of the solution with only two BSs will not exceed twice that of the optimal solution with countless BSs. Combining the conclusions in Propositions \ref{2BS_optimal} and \ref{SensingAccuracy}, it can be derived that the IRS tends to be associated with fewer BSs, thereby focusing the reflecting power to improve sensing performance, which is also evaluated via simulations in Section \ref{SimulationSection}. This conclusion offers valuable insights into optimizing the deployment density of BSs for positioning services. Specifically, it suggests that deploying sparse and well-distributed BSs in the reserved region is more cost-efficient.

\section{Multi-Vehicle Localization}
\label{MultiVehicle}
In this section, we propose a low-complexity resource allocation algorithm for multi-vehicle cases, and then present the performance analysis for the proposed algorithm.
\subsection{Algorithm Design}

Similar to the adopted beam flattening technique for single-vehicle cases, if $b_{k,m,n} = 1$, the phase shift of IRS $k$ is designed to cover the spatial span $\left[\bar \Phi_{k,m}, \underline \Phi_{k,m}\right]$ and $\left[\bar \Omega_{k,m}, \underline \Omega_{k,m}\right]$ along $x$- and $y$-axis, respectively. In this case, the echo signals will also be received by other BSs within this angular range relative to target $k$, resulting in an uncertain association between targets and measured distance. To tackle this issue, the interference caused by unassociated IRS should be avoided according to the spatial relationship between IRSs and BSs. Here, the interference relationship can be described according to the prior location of vehicles. Specifically, if BSs $m$ and IRS $k$ are associated at a certain time slot, and the adopted beam flattening techniques (c.f. Section \ref{SingleVehicle}) will cause interference to BS $m'$ if $\left[\bar \Phi_{k,m}, \underline \Phi_{k,m}\right] \cap \Phi_{k,m'} \ne \emptyset$ and/or $\left[\bar \Omega_{k,m}, \underline \Omega_{k,m}\right] \cap \Omega_{k,m'} \ne \emptyset$, and in this case, let $w_{k,m,m'} = 1$; otherwise, $w_{k,m,m'} = 0$. Notice that if $b_{k,m,n} = 1$, BS $m'$ should not work at the $n$th time slot if $w_{k,m,m'} = 1$. Based on the established interference graph $\{w_{k,m,m'}\}$, the interference can be avoided if the following constraints are satisfied, i.e.,
\begin{equation}\label{InterferenceAvoidance}
	b_{k,m,n} + \sum\nolimits_{k' \ne k}^K   b_{k',m',n} \le 2 - w_{k,m,m'}, \forall k, m, n.
\end{equation}
By satisfying the constraints in (\ref{InterferenceAvoidance}), the interference from unassociated IRS can be ignored due to the weak reflecting power outside the main lobe. Then, the SINR of the received echos can be expressed as
\begin{equation}\label{SensingReceivedPowerT}
	\gamma^{\mathrm{S}}_{k,m} = \sum\nolimits_{n = 1}^N b_{k,m,n}  \eta_{n} \bar \gamma_{k,m}^S,
\end{equation}
where $\bar \gamma_{k,m}^S = \frac{\beta^2_0 \Delta T L^2}{ Q_{k,m}^2 \Delta t  d_{k,m}^4 \sigma_s^2}$. Let $\tilde \gamma_{k,m}^S = \gamma_{k,m}^S \cos^2\phi_{k,m}$.
Then, the CRLB of target $k$ can be given by
\begin{equation}\label{MultiTarget_CRLB}
	{\rm{CRLB}}_k = \frac{ C_0 {\sum\nolimits_{m = 1}^M \tilde \gamma_{k,m}^S  }}{{\sum\nolimits_{j = 1}^{M-1} { \tilde \gamma_{k,j}^S  \left( {\sum\nolimits_{i = j+1  }^M {{{ \tilde \gamma_{k,i}^S  {\sin^2 (\varphi_{k,i} - \varphi_{k,j})}}}} } \right)} }}.
\end{equation}
Thus, (P1) can be equivalently transformed into 
\begin{alignat}{2}
	\label{P3}
	(\rm{P3}): \quad & \begin{array}{*{20}{c}}
		\mathop {\min }\limits_{{\bm{\eta}}, \{b_{k,m,n}\},N} \quad     \xi
	\end{array} & \\ 
	\mbox{s.t.}\quad
	& (\ref{P1}b)-(\ref{P1}e), \nonumber \\
	& {\rm{CRLB}}_k \le \xi, \forall k \in {\cal{K}}, \tag{\ref{P3}a} \\
	& b_{k,m,n} + \sum\nolimits_{k' \ne k}^M   b_{k',m',n} \le 2 - w_{k,m,m'}, \forall k , m, n.  & \tag{\ref{P3}b}
\end{alignat}
Note that under the optimal association $\{b^*_{k,m,n}\}$, it can be readily verified that the optimal time allocation $\{\eta_m\}$ can be solved by the MO-based algorithm presented in Section \ref{SingleVehicle}. Thus, a direct way to find the optimal solution to (P3) is to search all possible BS-IRS associations and calculate the corresponding optimal time allocation through the MO-based algorithm in Section \ref{OptimalFixDeployment}, and then the BS-IRS association with the minimum CRLB value is the optimal solution. However, such an operation is infeasible due to the high computational complexity, especially when the number of targets/BSs is large. On the other hand, it is difficult to jointly optimize target association $\{b_{k,m,n}\}$ and time allocation $\{\eta_n\}$ due to the non-convex constraints and closely coupled variables in ({\ref{P3}a}). To this end, we propose a heuristic BS-IRS association algorithm according to the derived conclusion in Section \ref{SingleVehicle}. 

According to the analysis in Section \ref{SingleVehicle}, the sensing performance of the solution with only two BSs approaches to that of the optimal solution. Thus, we propose a two-step algorithm to solve target association and time allocation, where only two BSs with the minimum CRLB value are selected to localize each vehicle, thereby reducing the total number of time slots and the algorithm complexity. In this case, the total number of BS-IRS associations is $2K$, and the probability of potential interference is reduced. The associated BSs with IRS $k$ are denoted by $k_1$ and $k_2$, and according to Proposition \ref{2BS_optimal}, the optimal ${\rm{CRLB}}^*_k$ is only affected by these two association time ratios, denoted by $\eta_{k,1}$ and $\eta_{k,2}$. Then, we normalize the resulting CRLB of each target, and the corresponding normalized time ratios can be denoted by $\bar \eta_{k,1}$ and $\bar \eta_{k,2}$, respectively, where $\bar \eta_{k,1} = {\rm{CRLB}}^*_k\eta_{k,1}$ and $\bar \eta_{k,2} = {\rm{CRLB}}^*_k\eta_{k,2}$. In this case, the CRLB is the same for each target using the normalized time ratio, thereby ensuring the fairness of multi-vehicle sensing. To make full use of space resources and achieve a better resource allocation effect, the BS-IRS associations are sorted according to the normalized time ratio, and then place in the time slots sequentially based on the interference-free graph $\{w_{k,m,m'}\}$ to avoid interference between different measurements. If the current BS-IRS association cannot be accommodated within the existing time slot set, a new time slot is generated, and the value of $N$ is incremented by one. After updating the time slot, the BS-IRS associations are sorted based on the expected length of the time slots, thereby enhancing utilization efficiency of time resources. The algorithm details are given in Algorithm \ref{LowComplexity}.

\begin{algorithm}[t]
	\caption{Low-Complexity Association Algorithm}
	\label{LowComplexity}
	\begin{algorithmic}[1]
		\STATE Input BS location, prior location of vehicles.
		\FOR {$k = 1:K$}
		\STATE Choose two BSs with minimum CRLB for target $k$, and obtain the optimal time allocation ratios according to Proposition \ref{2BS_optimal}, denoted by $\eta_{k,1}$ and $\eta_{k,2}$;
		\ENDFOR
		\WHILE {$2K$ BS-IRS associations}
		\STATE 
		Put the BS-IRS association results with fewer time ratio ($\min (\eta_{k,1},\eta_{k,2})$) into the first time slot. 
		\IF {$\sum\nolimits_{m' \ne m}^M w_{k,m,m'} =0$} 
		\STATE Continue to put it in.
		\ELSE
		\STATE $N \to N + 1$, put this association into the new added time slot.
		\STATE Sort the BS-IRS associations according to the normalized time ratio.
		\ENDIF
		\STATE After determining the target-BS association, the time allocation can be solved by the MO algorithm.
		\ENDWHILE
	\end{algorithmic}
\end{algorithm}

\subsection{Performance Analysis for Proposed Localization Scheme}
To draw useful insights into the proposed IRS-assisted localization scheme, in this subsection, we provide theoretical analysis to characterize the relationship between the sensing performance and the numbers of target/BSs/time slots ($K$/$M$/$N$), which is useful to facilitate the practical setup design for the proposed localization scheme. 

If $\sum\nolimits_{k = 1}^K \sum\nolimits_{m = 1}^M \sum\nolimits_{m' \ne m}^M w_{k,m,m'} =0$, this corresponds to the case that the location of BSs and vehicles is relatively dispersed, and any two BS-IRS association can exist simultaneously. In this case, the space resources can be fully utilized. On the other hand, if $\sum\nolimits_{k = 1}^K \sum\nolimits_{m = 1}^M \sum\nolimits_{m' \ne m}^M w_{k,m,m'}
 = KM(M-1)$, i.e., there is mutual interference in any pair of BS-IRS associations, then each measurement can only be performed in a time-division manner. For these two cases, the corresponding relationship between the minimum number of time slots and the maximum number of available targets is given as follows.

\begin{thm}\label{MinimumN_tmeslots}
	Given the number of BSs and targets, i.e., $M$ and $K$, the minimum number of time slots required to achieve effective target localization are given by 
	\begin{itemize}[leftmargin=*]
		\item $N_{\min} = \max(\left\lceil {\frac{2 K}{M} } \right\rceil, 2) $, if $\sum\nolimits_{k \!=\! 1}^K \sum\nolimits_{m \!=\! 1}^M \sum\nolimits_{m' \ne m}^M w_{k,m,m'} = 0$;
		\item $N_{\min} = {{2 K} } $, if $\sum\nolimits_{k = 1}^K \sum\nolimits_{m = 1}^M \sum\nolimits_{m' \ne m}^M w_{k,m,m'} = KM(M-1)$.
	\end{itemize}
\end{thm}
\begin{proof}
	First, we have $\sum\nolimits_{n = 1}^N \sum\nolimits_{m = 1}^M b_{k,m,n}  \ge 2$, i.e., the minimum number of BSs to successfully localize a target is two, and thus, the minimum total number of distance measurements satisfies $\sum\nolimits_{n = 1}^N \sum\nolimits_{k = 1}^K \sum\nolimits_{m = 1}^M b_{k,m,n} = 2 K $.
	Moreover, the maximum number of the associated targets (the number of distance measurements) at each time slot is less than the number of BSs and targets, i.e., $\sum\nolimits_{k = 1}^K \sum\nolimits_{m = 1}^M b_{k,m,n} \le \min(M, K)$. Hence, if there is no interference between different IRSs, i.e., $\sum\nolimits_{k = 1}^K \sum\nolimits_{m = 1}^M \sum\nolimits_{m' \ne m}^M w_{k,m,m'} = 0$, we have $\sum\nolimits_{k = 1}^K \sum\nolimits_{m = 1}^M b_{k,m,n} = M, \forall n$. Thus, in this case, the minimum number of time slots is $N_{\min} = \left\lceil {\frac{2 K}{M} } \right\rceil $. Similarly, if all targets interfere with each other, $\sum\nolimits_{k = 1}^K \sum\nolimits_{m = 1}^M \sum\nolimits_{m' \ne m}^M w_{k,m,m'} = KM(M-1)$, and they have to be sensed in a time-division manner. In this case, $N_{\min} = 2 K$.
\end{proof}

Lemma \ref{MinimumN_tmeslots} illustrates the minimum number of time slots under different interference cases. Intuitively, the less interference, the number of BSs can be exploited to reduce the number of slots, thereby reducing the complexity of positioning algorithm. 

\begin{thm}\label{MaximumK_tmeslots}
	Given the number of BSs and time slots, i.e., $M$ and $N$, the maximum number of targets that can be effectively localized is given by 
	\begin{itemize}[leftmargin=*]
		\item $K_{\max} \!=\! \left\lfloor {\frac{M N}{2}} \right\rfloor $, if $\sum\nolimits_{k \!= 1}^K \sum\nolimits_{m = 1}^M \sum\nolimits_{m' \ne m}^M w_{k,m,m'} =0$;
		\item $K_{\max} = \left\lfloor {\frac{N}{2}} \right\rfloor $, if $\sum\nolimits_{k = 1}^K \sum\nolimits_{m = 1}^M \sum\nolimits_{m' \ne m}^M w_{k,m,m'} = KM(M-1)$.
	\end{itemize}
\end{thm}
\begin{proof}
	This Lemma can be similarly proved as Lemma \ref{MinimumN_tmeslots} and the details are omitted for brevity.
\end{proof}

Lemmas \ref{MinimumN_tmeslots} and \ref{MaximumK_tmeslots} offer essential insights into the relationship between the number of targets, BSs, and time slots in the setup, which can help guide the parameter setting process effectively.

\begin{Pro}\label{FundamentalTradeoff}
	Under the same setup, the minimum CRLB in (P3) is non-increasing/decreasing as the number of targets/BSs decreases.
\end{Pro}
\begin{proof}
	First, the optimal solutions to (P3) with $K$ targets are denoted by $\{b^*_{k,m,n}\}$ and $\eta^*_n$, and the corresponding CRLB is denoted by ${\rm{CRLB}}^*$. With $K-1$ targets, it is not difficult to verify that the solutions $b^*_{k,m,n}$ and $\eta^*_n$ can also achieve the localization performance ${\rm{CRLB}}^*$ since the interference constraints can be satisfied. Similarly, it can be verified that the minimum CRLB in (P3) under $M$ BSs is no more than that under $M-1$ BSs.
\end{proof}

Proposition \ref{FundamentalTradeoff} reveals that the performance relationship between the deployment density of BSs, the number of targets to be sensed, and the localization performance requirements.

\section{Simulations}
\label{SimulationSection}	

To validate our analysis and characterize the performance of the proposed localization scheme, Monte Carlo simulation results are presented in this section. The system parameters are given as follows: $L_x = L_y = 40$, $K = 10$, $M = 4$, $\beta_0 = -30$dB, $\sigma^2_s = -80$dB, $P_{\mathrm{A}} = 1$W, $D_e = 5$m, $H_{BS} = 5$m, $H_{IRS} = 1$m, $\Delta T = 0.1$s, $\Delta t = 10^{-6}$s, and $C_0 = 0.1$. In addition, the following benchmark schemes are taken into account for comparison:
\begin{itemize}[leftmargin=*]
	\item {\bf{Average time allocation to all BSs (Average)}}: Each BS is assigned the same time ratio.
	\item {\bf{Closest two BSs (Closet)}}: The two closest BSs to each target are selected to be associated with the IRS.
	\item {\bf{Time division sensing (Time Division)}}: Each target is sensed on orthogonal time slots, with optimal time allocation separately.
\end{itemize}

\begin{table}
	\centering
	\caption{Statistics on the number of optimal associated BSs.}
	\label{Table1}
	\begin{tabular}{cccccc} 
		\toprule   
		Associated BS number & 1 & 2 & 3 & 4 & 5-10 \\
		\toprule    
		Proportion  & 0 & 90\% & 8\% & 2\% & 0 \\
		\toprule    
		Minimum value $\{\eta_m\}$  & - & 0.29 & 0.018 & 0.001 & - \\
		\toprule 
	\end{tabular}
\end{table}

\begin{figure*}[t]
	\centering
	\subfigure[Location of BSs and the target.]
	{	
		\label{figure7a}
		\includegraphics[width=7.2cm]{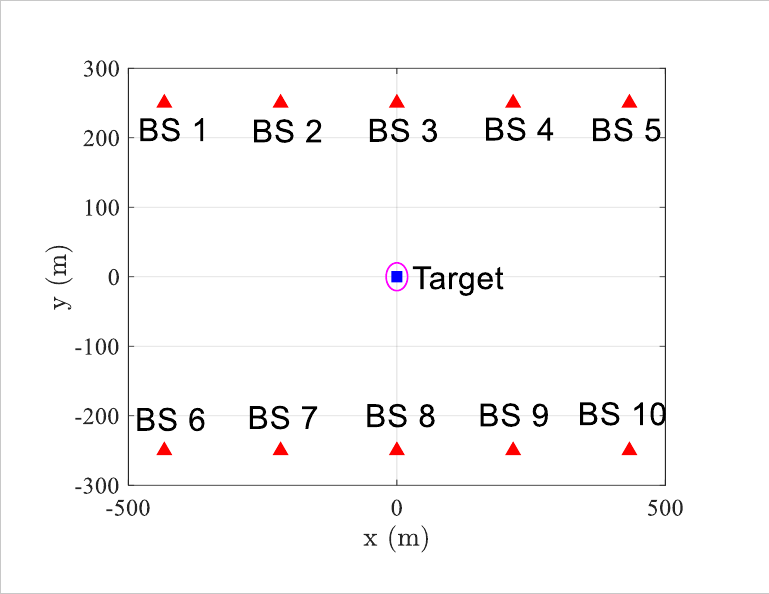}
	}
	\subfigure[Echo power at time slot 1.]
	{	
		\label{figure7b}
		\includegraphics[width=7.2cm]{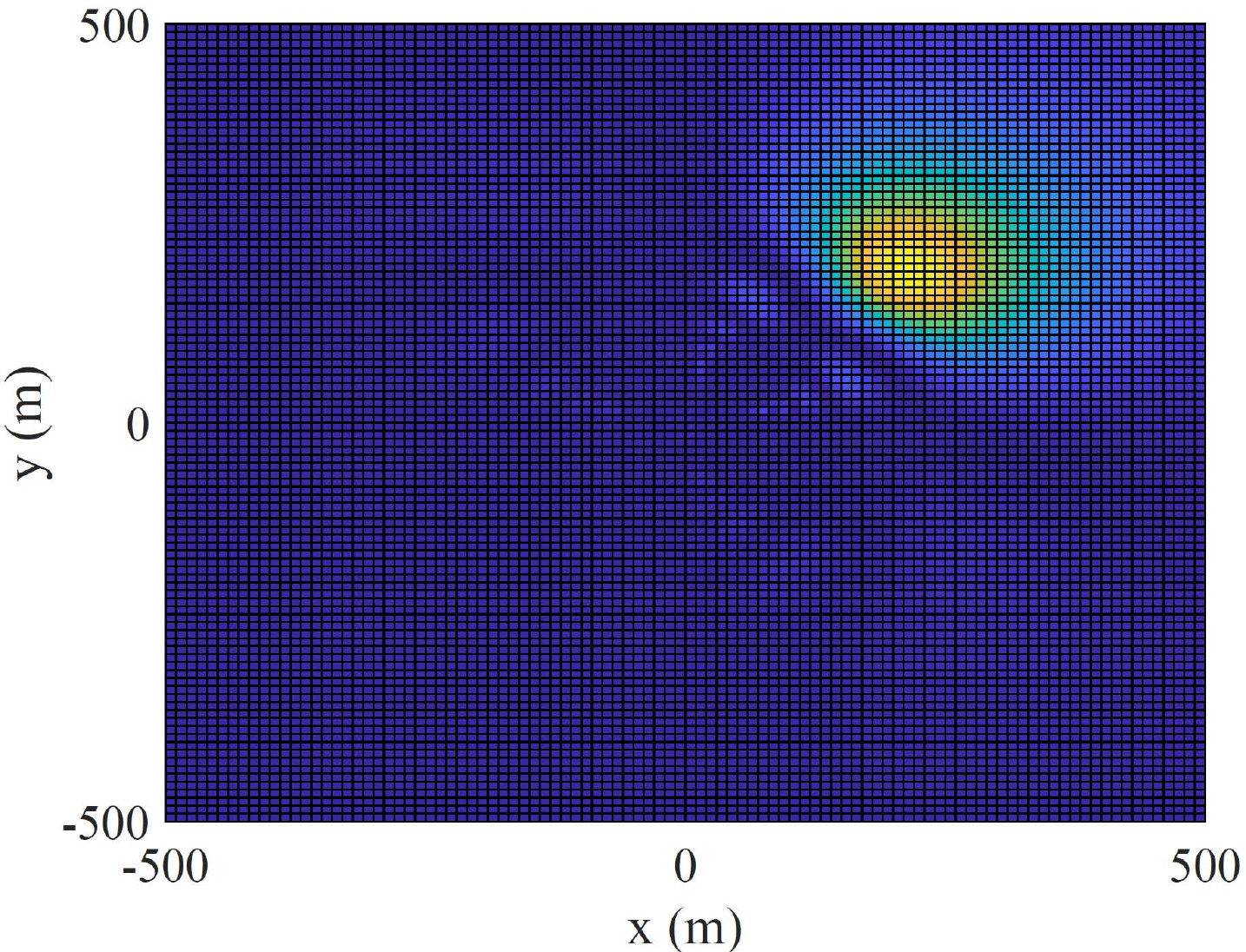}
	}
	\subfigure[Echo power at time slot 2.]
	{	
		\label{figure7c}
		\includegraphics[width=7.2cm]{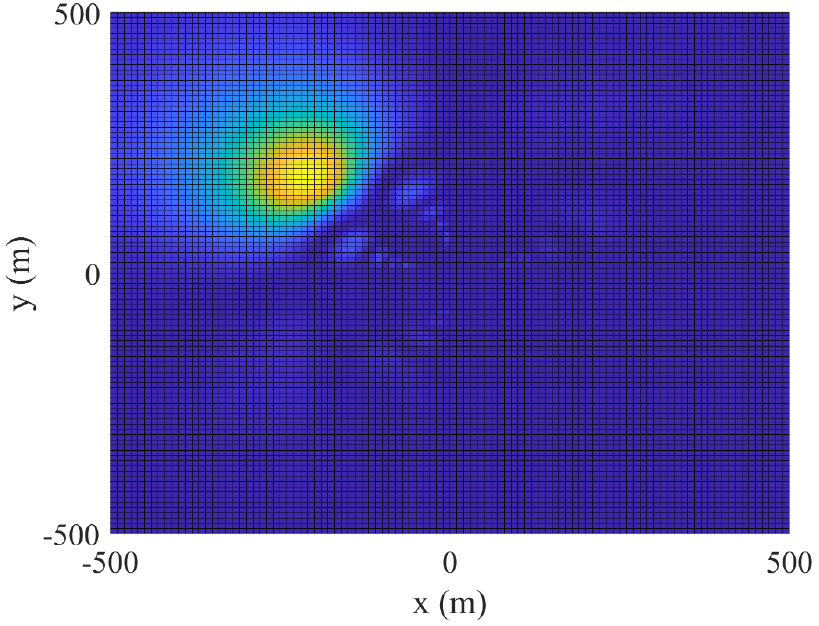}
	}
	\subfigure[Echo power at time slot 3.]
	{	
		\label{figure7d}
		\includegraphics[width=7.2cm]{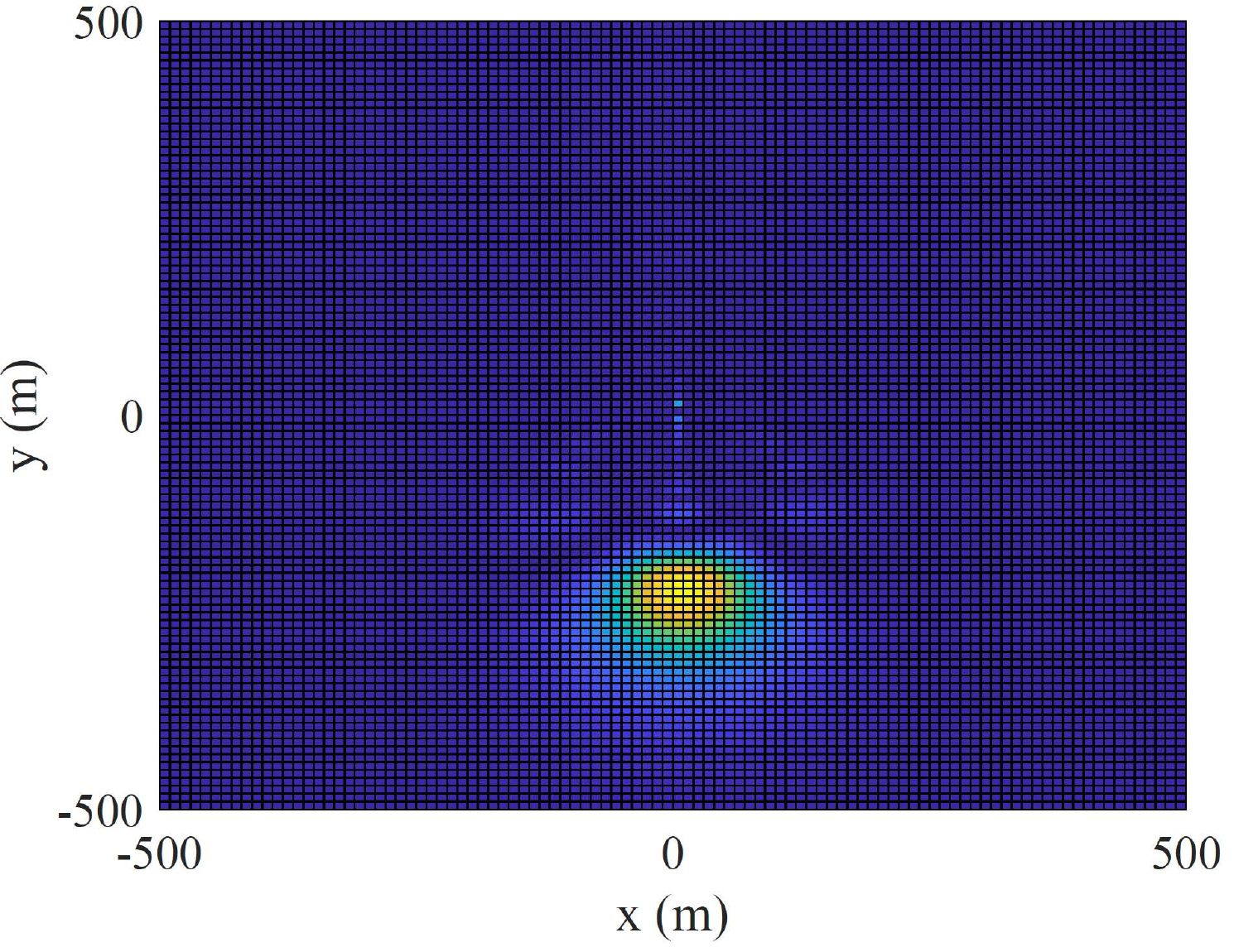}
	}		
	\caption{Echo signals power at different locations.}
	\label{figure7}
\end{figure*}

\subsection{Single-Vehicle Localization}

First, the number of the associated BSs for the optimal solution is statistically analyzed by setting $M = 10$, as shown in Table \ref{Table1}, where the location of BSs is randomly generated. It is worth noting that with probability 90$\%$, the IRS only associates with two BSs at the optimal solution, i.e., even if there are more candidate BSs, it tends to choose fewer BSs to associate with. Moreover, the optimal number of associated BSs is generally less than four. This is expected since associating with more BSs will disperse the echo's energy, resulting in a decrease in the average accuracy of distance measurements. In other words, the diversity gain of measurement directions cannot compensate for the loss of distance measurement accuracy. Also, this conclusion corroborates our analysis in Section \ref{FlexibleBS}, i.e., instances associated with a higher number of BSs are less likely to be optimally positioned. Moreover, when the number of associated BSs increases, the minimum time ratio of the optimal solution gradually decreases. When the optimal number of associated BSs is 4, the minimum time ratio is only 0.001, which indicates that the performance gain brought by the measurement of the corresponding BS is negligible. Therefore, in our proposed positioning scheme, it is sufficient for the IRS to actively reflect echo signals toward two or three BSs to improve sensing performance.

In Fig.~\ref{figure7}, the 2-dimensional (2D) power of echo signals is provided to illustrate the effectiveness of the control of echo signals. The location of 10 BSs and the IRS/target is shown in Fig.~\ref{figure7a}, where BSs and the IRS are represented by triangles and boxes, respectively, and the pink circle represents the uncertainty region of the vehicle's position. It is observed from Fig.~\ref{figure7} that only the associated BS 2, BS 4, and BS 8 are allocated an equal proportion of the time. As evident from the observation in Fig.~\ref{figure7}, the IRS is not obligated to be associated solely with the nearest BS. Conversely, the BS association takes into account both distance and orientation as determining factors.

\begin{figure}[t]
	\centering
	\includegraphics[width=8.3cm]{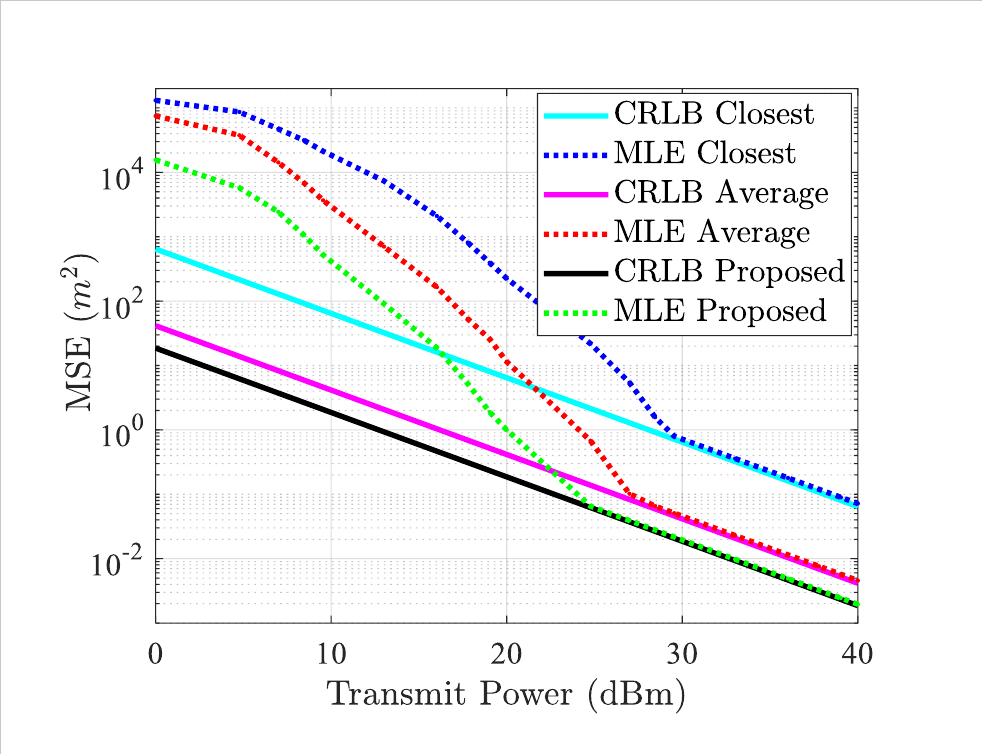}
	\caption{MSE comparison versus CRLB under different transmit power.}
	\label{figure4}
\end{figure}

Fig.~\ref{figure4} illustrates the performance for target localization in terms of the mean square error (MSE), with the increase of the transmit power $P_{\mathrm{A}}$. The MLE method is adopted for obtaining the target location via exhaustive grid search \cite{myung2003tutorial}. As expected, the location MSEs of these schemes are lower-bounded by the corresponding CRLB, and the CRLB is tight and can be achieved by the MLE in the high-transmit power regime, which is consistent with the analytical result in the literature \cite{Li2008RangeCompression, Liu2022CramerRaoBound}. It can be seen that the proposed method outperforms both benchmark scheme designs, especially when the transmit power is low. Compared to the 0.4/0.9$W$ transmit power's threshold to achieve a tight CRLB for the "average"/"closest" benchmark scheme, the proposed scheme can take the optimal time ratio to achieve that at a lower transmit power (i.e., 0.2$W$) due to a better balance between different distance measurements. This proves that CRLB minimization for the proposed localization scheme is effective to improve the target estimation performance.

\begin{figure}[t]
	\centering
	\includegraphics[width=8.3cm]{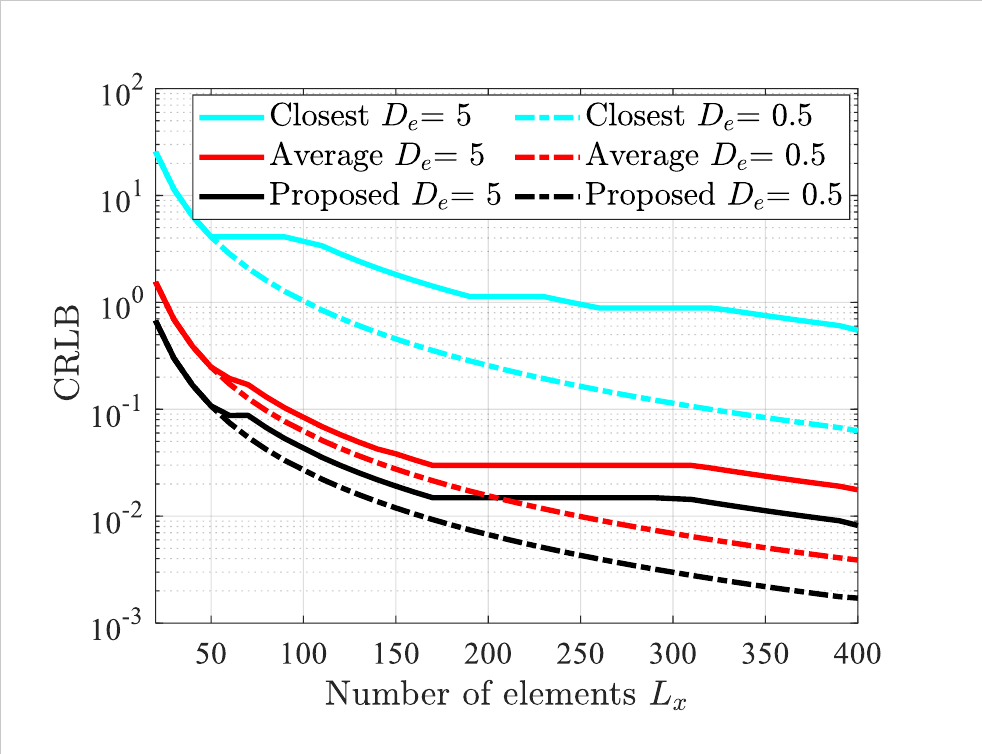}
	\caption{Localization performance comparison versus different numbers of IRS elements.}
	\label{figure5}
\end{figure}

In Fig.~\ref{figure5}, the localization performance is compared under different numbers of IRS elements with $P_{\mathrm{A}} = 0.1$W. As shown in Fig.~\ref{figure5}, under different IRS elements, the "average" scheme and the "closest" scheme lead to 1.3 and 37.3 higher CRLB as compared to the proposed scheme, respectively. Also, the CRLB gap between the case with $D_e = 0.5$m and the case with $D_e = 5$m becomes more pronounced with increasing the number of IRS elements $L$. The main reason is that the beam width of the echo signals becomes narrower with the increasing $L$, and thus the IRS needs to be divided into more sub-groups to ensure that the BS can receive echo signals in the uncertain prior region. The "closest" scheme with the prior location's error $D_e = 5$m achieves much worse positioning performance as compared to that of $D_e=0.5$m, which is mainly because the closer the potential distance between BS and IRS, the wider the beam should be designed to cover the angular range. For $D_e = 5$m, the positioning performance exhibits incremental improvements in a stepwise manner due to the IRS only being divided into integer parts.

\begin{figure}[t]
	\centering
	\includegraphics[width=8.3cm]{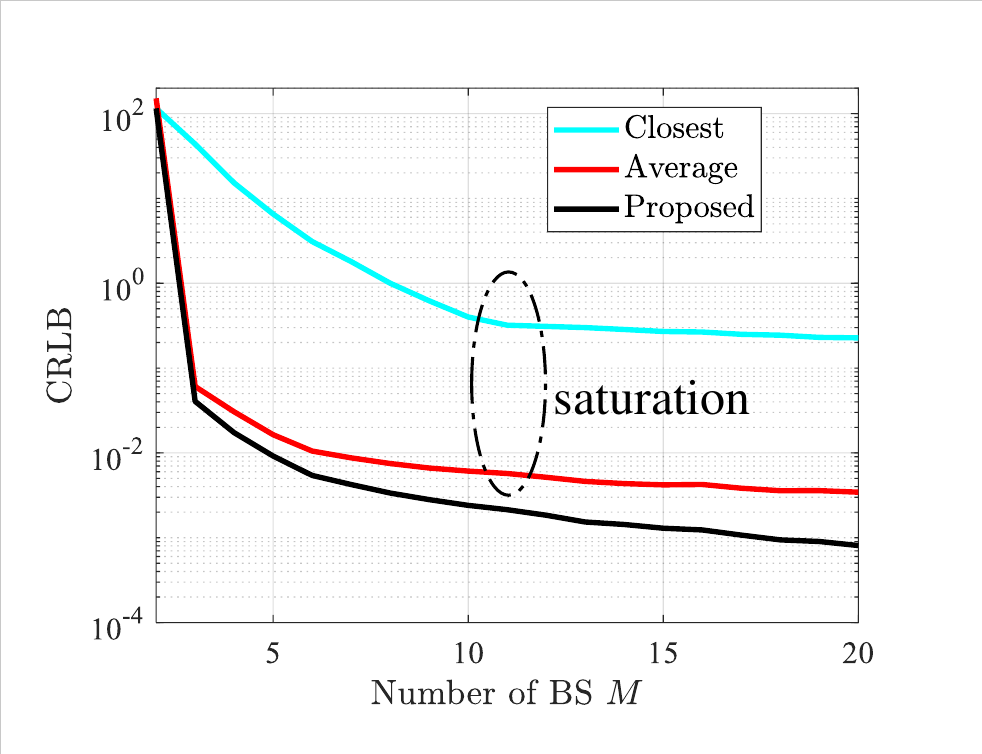}
	\caption{Statistical CRLB comparison versus different numbers of BSs.}
	\label{figure6}
\end{figure}

The statistical localization performance is compared in Fig.~\ref{figure6} for different numbers of BSs $M$, with $P_{\mathrm{A}} = 0.1$W. When the number of BSs increases from 2 to 5, the CRLB of these considered schemes is significantly decreased since there are more BSs in different azimuths to achieve a better geometry gain or an increasing SNR of echo signals. Also, as the number of candidate BSs increases, the performance improvement of our proposed algorithm over the "average" scheme gradually increases. This suggests that the benefit of associating to more BSs for higher diversity does not outweigh the SNR loss of each measurement, thereby reducing the sensing accuracy, especially for IRSs with more elements. It can be observed that the positioning performance reaches saturation when the number of BSs is larger than 11, due to the non-optimal time resource allocation for the "average" scheme or non-optimal association for the "closest" scheme, respectively.

\subsection{Multiple-Vehicle Localization}

In this subsection, the simulation results under the multi-vehicle setup are provided to analyze the relationship between the numbers of time slots, BSs, and targets. In Fig. \ref{figure8a}, we show the statistical average number of time slots $N$ comparisons versus the number of BSs by setting $K = 10$. It can be seen that the numbers of time slots for both the proposed scheme and the "closest" scheme decrease as the number of BSs increases, which is consistent with the analysis in Proposition \ref{SensingAccuracy}. The main reason is that there are more potential associations in BS-IRS association candidates with no interference, thus reducing the required number of time slots and improving the utilization effectiveness of time resources. Moreover, the number of time slots of our proposed scheme is significantly reduced as compared to the "time division" scheme under higher BS density, since more BS-IRS associations can be established at the same time to perform distance measurements simultaneously. Differently, the number of time slots of the "time division" scheme increases slightly as the number of BSs increases, since the optimal number of associated BSs for this scheme increases with the BS density to improve the localization performance. It can be found that under a larger error bound of the vehicle's prior location, the optimized number of time slots $N$ is reduced more significantly as compared to that under a lower error bound. 

\begin{figure}[t]
	\centering
	\includegraphics[width=8.3cm]{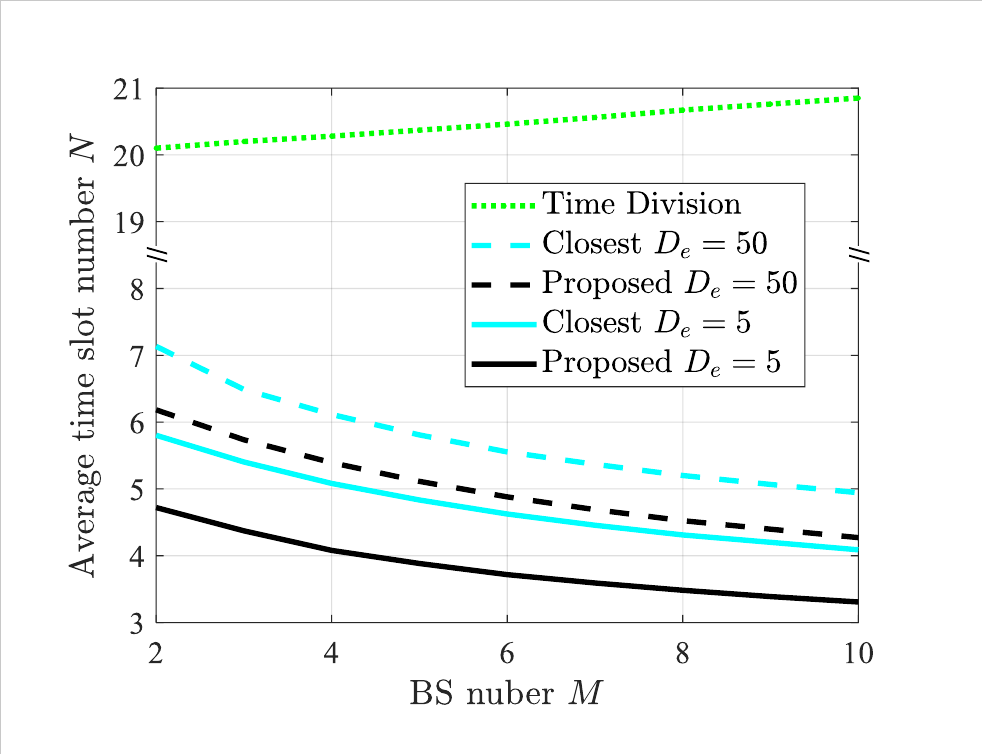}
	\caption{Statistical time slots comparison versus different numbers of BSs.}
	\label{figure8a}
\end{figure}

\begin{figure}[t]
	\centering
	\includegraphics[width=8.3cm]{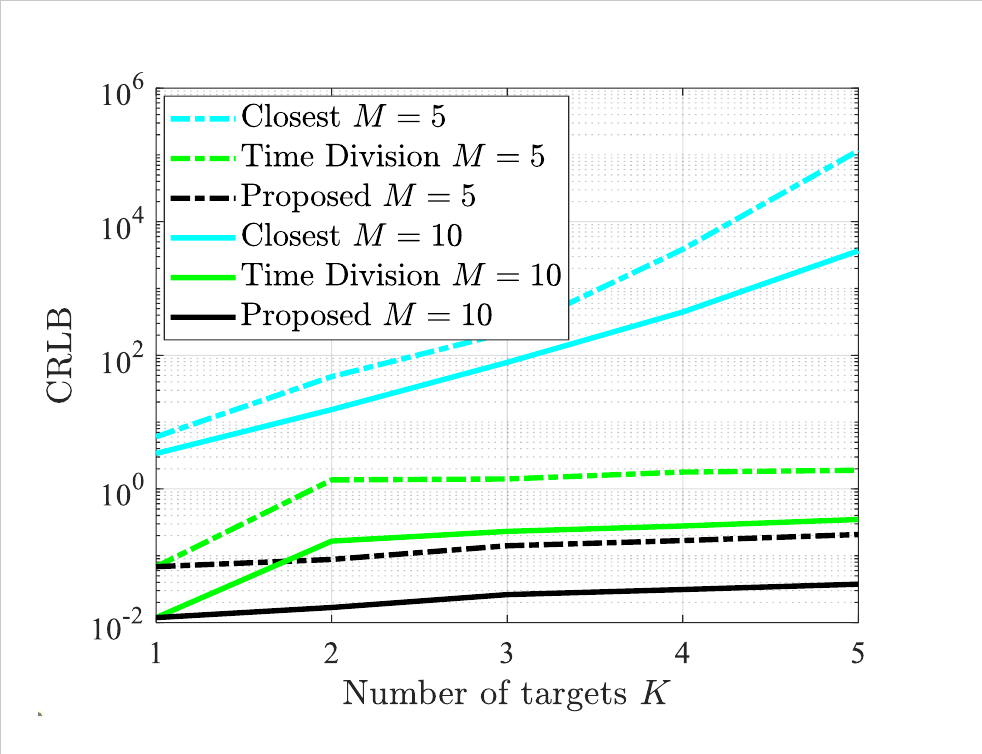}
	\caption{Statistical CRLB comparison versus different numbers of targets.}
	\label{figure8b}
\end{figure}

In Fig.~\ref{figure8b}, it can be observed that the CRLB of these considered schemes increases monotonically with the number of targets $K$, and in particular, the reduction of the CRLB achieved by our proposed scheme over two benchmark schemes increases as the number of targets increases. For a given number of targets, the CRLB with $M = 10$ can be reduced about 5 times compared to that with $M = 5$, and the main reason for the reduction in CRLB is that there are more potential BS-IRS associations with no interference, ensuring more distance measurements can be obtained simultaneously. For the "closest" scheme, when the number of BSs is less (i.e., $M = 5$), the CRLB increases faster caused by more severe interference and the reduced utilization efficiency of time resources. Notice that the proposed scheme is reduced to the "time division" scheme for single-target cases, while for two-target cases, the positioning performance of the "time division" scheme is severely degraded as compared to the proposed scheme. The main reason is that the interference probability between different BS-IRS associations is relatively low for less-target cases, and the "time division" scheme does not make full use of time and space resources. As the number of targets increases, the proposed scheme can bring greater positioning performance gain compared to the "closest" scheme. The main reason is that as the number of targets increases, the "closest" scheme only considers associating to the closest BSs, which inevitably increases the probability that different targets are associated with the same BSs, thus requiring more time slots for non-interference sensing between different targets.

\begin{figure}[t]
	\centering
	\includegraphics[width=8.3cm]{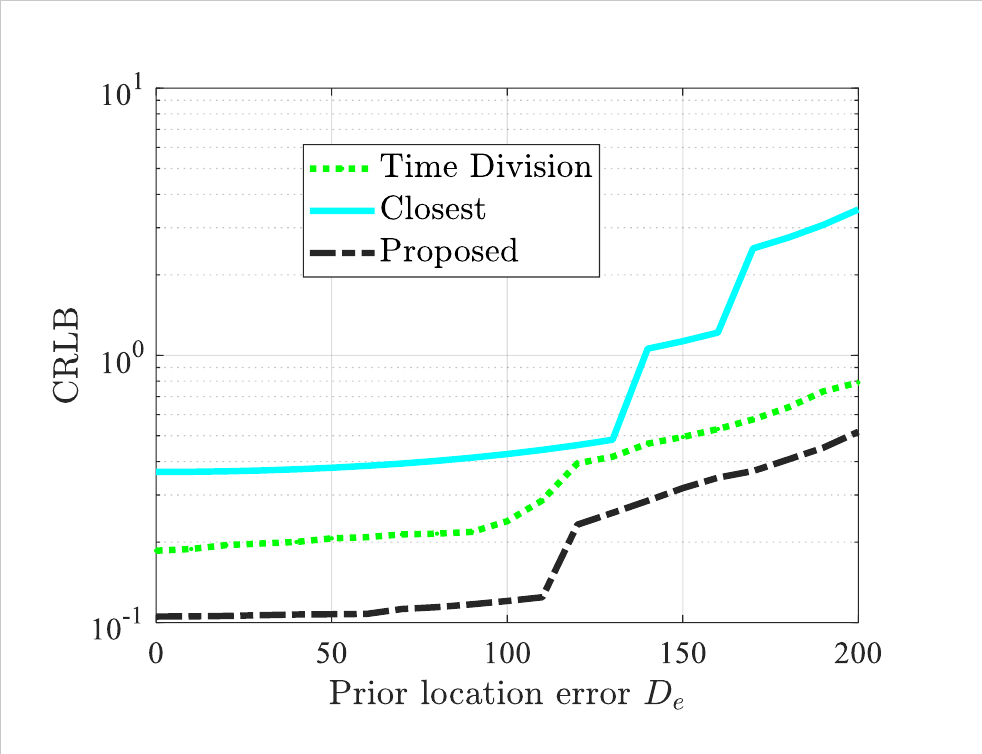}
	\caption{Statistical CRLB comparison versus different prior location error $D_e$.}
	\label{figure11}
\end{figure}

Fig.~\ref{figure11} shows the statistical CRLB comparison versus different prior location error $D_e$. It can be observed that the CRLB increases monotonically with the increase of the uncertainty range $D_e$, i.e., the positioning performance of the scheme decreases. We note that with $D_e = 0$m, the minimum CRLB is higher than 0 due to the antenna noise and the resulting limited SNR. When the error range increases to 120$m$, the performance difference between the proposed scheme and the “time division” scheme decreases. The main reason is that as the error range increases, the likelihood of interference between multiple associations also increases. Consequently, fewer BS-IRS associations can be executed within the same time slot. In this case, the proposed scheme tends to adopt a time division strategy to mitigate interference and ensure optimal performance. Moreover, to achieve the required sensing performance CRLB, the proposed scheme can greatly improve the ability to tolerate prior position errors, e.g., when ${\rm{CRLB}} = 0.3$, the largest available range of $D_e$ of our proposed scheme is respectively 0.7 and 1.5 times higher as compared to that of the "time division" scheme and the "closest" scheme. 

\section{Conclusions}
In this paper, the CRLB-minimization problem for IRS-mounted target localization systems is studied and a novel cooperative localization scheme with the active control of echo signals is proposed. In our considered system, the target association, the IRS phase shifts, and the dwell time are jointly optimized to improve sensing performance. For a single-target setup, we derive the optimal BS location to provide a lower performance bound of the original problem. Then, we prove that the transformed problem is an MO, which can be optimally solved by the Polyblock-based algorithm. Furthermore, we theoretically analyzed the relationship between the number of time slots, BSs, and targets, providing useful guidance for the practical implementation of the proposed localization scheme. Finally, simulation results demonstrate that our proposed scheme achieves a significantly lower localization error over the benchmark schemes and that the effective association between IRSs and BSs can greatly reduce the BSs'/time slots' number requirements. 

\normalsize 
\section*{Appendix A: \textsc{Proof of Lemma \ref{LowerBoundCRLB}}}

First, it is not difficult to verify that CRLB is a monotonically decreasing function of $\phi_m$. Thus, we have $\phi_m^* = \arcsin{\frac{|H_{BS} - H_{IRS}|}{\underline d}}$ at the optimal BS location.
According to the definition in (\ref{DefinitionCRLB}), we have inequality
\begin{equation}\label{EquationProof}
	\begin{aligned}
		{\rm{CRLB}} \ge \frac{{\sum\nolimits_{m = 1}^M {\frac{{{{ {\cos^2 {\phi _m}} }}}}{{\sigma _m^2}}} }}{{\sum\nolimits_{m = 1}^M {\frac{{{{ {\cos^2 {\varphi _m}\cos^2 {\phi _m}} }}}}{{\sigma _m^2}}} \sum\nolimits_{m = 1}^M {\frac{{{{ {\sin^2 {\varphi _m}\cos^2 {\phi _m}} }}}}{{\sigma _m^2}}} }}.
	\end{aligned}
\end{equation}
The equality in (\ref{EquationProof}) holds if and only if  $\sum\nolimits_{m = 1}^M {\frac{{\cos {\varphi _m}\cos {\phi _m}\sin {\varphi _m}\cos {\phi _m}}}{{\sigma _m^2}}} = 0$. Let
\begin{equation}
	{\sum\nolimits_{m = 1}^M {\frac{{{{\left( {\sin {\varphi _m}} \cos \phi_m \right)}^2}}}{{\sigma _m^2}} = G} }.
\end{equation}
Then, the denominator of (\ref{EquationProof}) can be rewritten as follows:
\begin{align}\label{MaximumDenominator}
	&\sum\nolimits_{m = 1}^M {\frac{{{{\left( {\sin {\varphi _m}\cos {\phi _m}} \right)}^2}}}{{\sigma _m^2}}} \sum\nolimits_{m = 1}^M {\frac{{{{\left( {\cos {\varphi _m}\cos {\phi _m}} \right)}^2}}}{{\sigma _m^2}}} \nonumber  \\
	=& \left( {\sum\nolimits_{m = 1}^M {\frac{{{{\left( {\cos {\phi _m}} \right)}^2}}}{{\sigma _m^2}} - G} } \right)G.
\end{align}
The value in (\ref{MaximumDenominator}) is maximized when ${G^*} = \frac{1}{2}\sum\nolimits_{m = 1}^M {\frac{{{{ {\cos^2 {\phi _m}} }}}}{{\sigma _m^2}}}$. Hence, we have
	\begin{align}
			&{\rm{CRLB}} = \frac{{\sum\nolimits_{m = 1}^M {\frac{{{{ {\cos^2 {\phi _m}} }}}}{{\sigma _m^2}}} }}{{\left( {\sum\nolimits_{m = 1}^M {\frac{{{{ {\cos^2 {\phi _m}} }}}}{{\sigma _m^2}} - G} } \right)G}} \nonumber \\
			 \overset{(b)}{\ge} &  \frac{{\sum\nolimits_{m = 1}^M {\frac{{{{ {\cos^2 {\phi _m}} }}}}{{\sigma _m^2}}} }}{{\left( {\sum\nolimits_{m = 1}^M {\frac{{{{ {\cos ^2 {\phi _m}} }}}}{{\sigma _m^2}} - \frac{1}{2}\sum\nolimits_{m = 1}^M {\frac{{{{ {\cos^2 {\phi _m}} }}}}{{\sigma _m^2}}} } } \right)\frac{1}{2}\sum\nolimits_{m = 1}^M {\frac{{{{ {\cos^2 {\phi _m}} }}}}{{\sigma _m^2}}} }} \nonumber \\
		=  & \frac{4}{{\sum\nolimits_{m = 1}^M {\frac{{{{ {\cos^2 {\phi _m}} }}}}{{\sigma _m^2}}} }}  =  \frac{4 C_0}{{\sum\nolimits_{m = 1}^M {{{{{\eta_m \bar \gamma_m^S  }}}}} }} \overset{(c)}{\ge}  \frac{4 C_0 \Delta t  \underline d^4 \sigma_s^2}{ \Delta T \beta^2_0 L^2 }.
	\end{align}
The equality in ($b$) holds when $\sum\nolimits_{m = 1}^M {\frac{{{{ {\sin^2 {\varphi _m}} \cos^2 \phi_m }}}}{{\sigma _m^2}}} = \frac{1}{2}\sum\nolimits_{m = 1}^M {\frac{{{{ {\cos^2 {\phi _m}} }}}}{{\sigma _m^2}}} $. Therefore, an optimal solution that can achieve the lower bound in Lemma \ref{LowerBoundCRLB} is $\varphi_m^* = \frac{2(m-1)\pi}{M}$ and $\phi_m^* = \arcsin{\frac{|H_{BS} - H_{IRS}|}{\underline d}}, m = 1, \cdots, M$ under the constraints $\sum\nolimits_{m = 1}^M {\frac{{\cos {\varphi _m}\cos {\phi _m}\sin {\varphi _m}\cos {\phi _m}}}{{\sigma _m^2}}} = 0$ and $\sum\nolimits_{m = 1}^M {\frac{{{{ {\sin^2 {\varphi _m}} \cos^2 \phi_m }}}}{{\sigma _m^2}}} = \frac{1}{2}\sum\nolimits_{m = 1}^M {\frac{{{{ {\cos^2 {\phi _m}} }}}}{{\sigma _m^2}}} $.
By combining the above analysis, the minimum ${\rm{CRLB}}  = \frac{4 C_0 \Delta t  \underline d^6 \sigma_s^2}{ \Delta T \beta^2_0 L^2  \left(\underline d^2 - (H_{BS} - H_{IRS})^2\right)}$.

\section*{Appendix B: \textsc{Proof of Proposition \ref{MonoProof}}}
First, we let $x_m = {\eta_m} \bar \gamma^S_m \cos^2 \phi_m$. Proposition \ref{MonoProof} holds if CRLB is a monotonically decreasing function of $x_m$ since $x_m$ is an affine transformation of $\eta_m$. Let $${\rm{CRLB}} = \frac{{\sum\nolimits_{m = 1}^M {{x_m}} }}{{\sum\nolimits_{m = 1}^{M - 1} {{x_m}\sum\nolimits_{i = m}^M {{x_i}{{ {\sin^2 \left( {{\phi _i} - \sin {\phi _m}} \right)} }}} } }}  \buildrel \Delta \over = f({\bm{x}}),$$ where ${\bm{x}} = [x_1, \cdots, x_M]$.
Then, the partial derivative of {{CRLB}} with respect to $x_{m'}$ is given by $\frac{{\partial f({\bm{x}})}}{{\partial x_{m'}}} = \frac{g(\bm{x})}{{{{\left( {\sum\nolimits_{m = 1}^{M \!-\! 1} {{x_m}\sum\nolimits_{i = m}^M \! {{x_i}{{ {\sin^2 \left( {{\phi _i} -  {\phi _m}} \right)} }}} } } \right)}^2}}}$, where $g(\bm{x}) = \sum\nolimits_{m \ne m'}^{M - 1} {{x_{m}}\sum\nolimits_{i \ne m,m'}^M {{x_i}{{ {\sin^2 \left( {{\phi _i} -  {\phi _{m}}} \right)} }}} }  - \left( {\sum\nolimits_{m \ne m'}^M {{x_{m}}} } \right)\sum\nolimits_{i \ne {m'}}^M {{x_i}{{ {\sin ^2 \left( {{\phi _i} -  {\phi _{m}}} \right)} }}}$. In the following, it will be proved that for any given $\{\phi_m\}$, we have $\frac{{\partial^2 f({\bm{x}})}}{{\partial^2 x_{m'}}} \le 0$. Since $\frac{\partial^2 g(\bm{x})}{\partial^2 x_m} \le 0$, $x^*_m =  \arg \max \limits_{x_m} g(\bm{x}) = 0$ if $\frac{\partial g(\bm{x})}{\partial x_m} \le 0$. In this case, $x_m$ can be removed. Otherwise, $\frac{\partial g(\bm{x})}{\partial x_m} > 0$, the optimal solution to the maximum value of $g(\bm{x})$ is obtained if and only if $\frac{\partial g(\bm{x})}{\partial x_m} = 0$. In this case, we have 
	\begin{align}
		&g({\bm{x}}) \nonumber \\
		= \! & \sum\nolimits_{m \ne m'}^{M} \!  x_m \!  \sum\nolimits_{i \! \ne \!  m,m'}^{M} ( \sin^2(\phi_m \! - \!  \phi_{m'}) \!  - \!  \sin^2\phi_m \!  - \!  \sin^2\phi_{m'} )  \nonumber \\
		&- ( \sum\nolimits_{m \ne m'}^{M} x_m^2 \sin^2(\phi_m) ) \le 0.
	\end{align}
Thus, $\frac{{\partial f({\bm{x}})}}{{\partial x_{m}}} \le 0$, $\forall m \in {\cal{M}}$. By combining the above results, the proof is completed.

\footnotesize  	
\bibliography{mybibfile}
\bibliographystyle{IEEEtran}

\end{document}